\documentclass{amsart}[11pt]
\usepackage{amsmath, amssymb, amsthm, amsfonts}
\usepackage{mathrsfs}
\usepackage[margin=1in]{geometry}
\usepackage{color, colortbl}
\usepackage{enumerate}
\usepackage{graphicx}
\usepackage{multirow}
\usepackage{hyperref}
\usepackage{enumerate}
\usepackage{tikz}
\usepackage{subfigure}
\usepackage{float}

\usetikzlibrary{decorations.pathreplacing, positioning}

\newtheorem{theorem}{Theorem}[section]

\theoremstyle{definition}
\newtheorem{definition}[theorem]{Definition}
\newtheorem{remark}[theorem]{Remark}
\newtheorem{assumption}[theorem]{Assumption}

\newtheorem{example}[theorem]{Example}
\numberwithin{equation}{section}

\definecolor{Gray}{gray}{0.9}
\newcommand{\ind}{1\hspace{-2.1mm}{1}}
\newcommand{\I}{\mathtt{i}}
\newcommand{\CC}{\mathbb{C}}
\newcommand{\RR}{\mathbb{R}}
\newcommand{\PP}{\mathbb{P}}
\newcommand{\EE}{\mathbb{E}}
\newcommand{\Lm}{\mathrm{L}}
\newcommand{\D}{\mathrm{d}}
\newcommand{\Oo}{\mathcal{O}}
\newcommand{\Dd}{\mathcal{D}}

\newcommand{\Cc}{\mathcal{C}}

\newcommand{\Ll}{\mathcal{L}}
\newcommand{\la}{\left\langle}
\newcommand{\ra}{\right\rangle}
\newcommand{\dt}{\partial_{t}}
\newcommand{\dxx}{\partial^2_{x}}

\newcommand{\dxom}{\partial_{x,\omega}}
\newcommand{\dx}{\partial_{x}}
\newcommand{\dOm}{\partial_{\Om}}

\newcommand{\dom}{\partial_{\omega}}
\newcommand{\domom}{\partial^{2}_{\omega}}
\newcommand{\E}{\mathrm{e}}
\newcommand{\Ps}{\mathrm{P}}
\newcommand{\PHat}{\widehat{\Ps}}
\newcommand{\QQ}{\mathbb{Q}}
\newcommand{\Xm}{\mathbf{X}}
\newcommand{\Zm}{\mathbf{Z}}
\newcommand{\Bm}{\mathbf{B}}
\newcommand{\Wm}{\mathbf{W}}
\newcommand{\Kr}{\mathrm{K}}
\newcommand{\lr}{\mathfrak{C}}
\newcommand{\Lr}{\mathfrak{L}}
\newcommand{\bm}{\mathbf{b}}
\newcommand{\Om}{\boldsymbol \omega}
\newcommand{\dist}{\boldsymbol \D}

\newcommand{\TTheta}{\boldsymbol \Theta}
\newcommand{\psib}{\boldsymbol \psi}
\newcommand{\varpb}{\boldsymbol \varphi}
\newcommand{\sigl}{\sigma_{L}}
\newcommand{\Sigb}{\boldsymbol \Sigma}
\newcommand{\sm}{\boldsymbol \sigma}
\newcommand{\ddelta}{\boldsymbol \delta}
\newcommand{\wwb}{\boldsymbol{\mathrm{w}}}
\newcommand{\kkappa}{\boldsymbol{\kappa}}
\newcommand{\ot}{\otimes^{t}}

\newcommand{\Aa}{\mathcal{A}}
\newcommand{\Ff}{\mathcal{F}}
\newcommand{\Dk}{\mathfrak{D}}
\newcommand{\Ik}{\mathfrak{I}}
\newcommand{\Fk}{\mathfrak{F}}
\newcommand{\rrho}{\overline{\rho}}
\newcommand{\Lo}{\overline{\Lambda}}
\newcommand{\eps}{\varepsilon}
\newcommand{\Th}{\thetab}
\newcommand{\thetab}{\boldsymbol{\theta}}

\newcommand{\half}{\frac{1}{2}}

\begin{document}

\title{Deep Curve-dependent PDEs for affine rough volatility}
\author{Antoine Jacquier}
\address{Department of Mathematics, Imperial College London, and Alan Turing Institute}
\email{a.jacquier@imperial.ac.uk}
\author{Mugad Oumgari}
\address{Lloyds Banking Group plc, Commercial Banking, 10 Gresham Street, London, EC2V 7AE,~UK}
\email{Mugad.Oumgari@lloydsbanking.com}
\thanks{The views and opinions expressed here are the authors' and do not represent the opinions of their employers. 
They are not responsible for any use that may be made of these contents. No part of this
presentation is intended to influence investment decisions or promote any product or service.
The authors would like to thank Lukasz Szpruch and Bernhard Hientzsch for stimulating discussions, as well as the referees 
and the Associate Editor for their insightful comments.}
\date{\today}

\begin{abstract}
We introduce a new deep-learning based algorithm to evaluate options in affine rough stochastic volatility models.
Viewing the pricing function as the solution to a curve-dependent PDE (CPDE), 
depending on forward curves rather than the whole path of the process, 
for which we develop a numerical scheme based on deep learning techniques.
Numerical simulations suggest that the latter is a promising alternative to classical Monte Carlo simulations.
\end{abstract}

\maketitle

\section{Introduction}
Stochastic models in financial modelling have undergone many transformations since the Black and Scholes 
model~\cite{Black}, and its most recent revolution, pioneered by Gatheral, Jaisson and Rosenbaum~\cite
{Volrough}, has introduced the concept of rough volatility.
In this setting, the instantaneous volatility is the solution to a stochastic differential equation driven by 
a fractional Brownian motion with small (less than a half) Hurst exponent, synonym of low H\"older regularity 
of the paths.
Not only is this feature consistent with historical time series~\cite{Volrough}, but it further allows to 
capture the notoriously steep at-the-money skew of Equity options, as highlighted in~\cite{AlosLeon, BFGHS, 
BFG15, Forde, Fukasawa}. 
Since then, a lot of effort has been devoted to advocating this new class of models and to showing the full 
extent of their capabilities, in particular as accurate dynamics for a large class of assets~\cite{BLP16}, and 
for consistent pricing of volatility indices~\cite{HJT18, JMM17}.
Nothing comes for free though, and the flip side of this new paradigm is its computational cost.
With the notable exception of the rough Heston model~\cite{AbiJ, EEFR, Roughening, ERHedging, ERChar}
and its affine extensions~\cite{AffineVolterra, AffineFwd}, the absence of Markovianity of the fractional 
Brownian motion prevents any pricing tools other than Monte Carlo simulations;
the simulation of continuous Gaussian processes, including fractional Brownian motion, is traditionally slow 
as soon as one steps away from the standard Brownian motion.
However, the clear superiority--for estimation and calibration--of these rough volatility models has 
encouraged deep and fast innovations in numerical methods for pricing, 
in particular the now standard Hybrid scheme~\cite{BLP15, HPV17}
as well as Donsker-type theorems~\cite{HJMu, Parczewski}, numerical approximations~\cite{Bayer1, Harms, 
Rados} and machine learning-based techniques~\cite{DeepL, Stone}.

In fact, industry practice is often entrenched, not in pure stochastic volatility models, 
but in models enhanced with a local volatility component, \`a la Dupire~\cite{DupireVol}, 
thereby ensuring an exact fit to the observed European option price surface.
The natural next step for rough volatility models is to include such a component, 
which we do here.
Pricing in such model is usually performed by simulation, but we adopt a different strategy, 
following the recent development by Viens and Zhang~\cite{Viens}, 
who proved an analogous version of the Feynman-Kac theorem for rough volatility models.
The fundamental difference is that the corresponding partial differential equation
is now path-dependent. 
This forces us to revisit classical market completeness results in the setting of 
rough local stochastic volatility models.
Path-dependent partial differential equations (PPDEs) have been studied extensively by Touzi, Zhang and co-authors~\cite{Touzi3, Touzi2, Touzi1}, and we shall draw existence and uniqueness from their works.
Despite these advances, though, very little has been developed to solve these PPDEs numerically, 
with the sole exception of a path-dependent version of the Barles and Souganidis' monotone scheme~\cite
{Barles} by Zhang and Zhuo~\cite{Zhang} and Ren and Tan~\cite{RenTan}.
The implementation thereof is however far from obvious.
In the context of rough volatility--at least for a certain subclass with a specific structure (see Assumptions~\ref{assu:Martingale}(ii) and~\ref{assu:Functional}), that includes the rough Heston and the rough Bergomi models--and for European vanilla options, we show that the pricing function depends on forward-looking curves rather than the whole path (including the past) of the process.
We therefore consider a particular subclass of path-dependent PDEs, that we call curve-dependent PDEs (CPDEs), 
for which we develop a novel algorithm based on the discretisation of the pricing CPDE, 
which we then solve using a deep-learning-based backward algorithm inspired by the tec	hnique pioneered by E, Han and Jentzen~\cite{JentzenDeep}.
Our work focuses on European option on the underlying stock price, in which case we know, mainly thanks to~\cite{Bergomi, Guyon},
that, viewed as forward variance curve models, the price is indeed related to a curve-dependent problem.
However, we show that for other options, the whole path is in fact needed, 
justifying our use of the Viens-Zhang approach~\cite{Viens}.
There is a tight link, obviously, between these SDEs and SDEs on Hilbert spaces, 
viewed as equations in infinite dimensions and we refer the interested reader to~\cite{DaPrato} for more details 
and to~\cite{Cont} for early applications thereof in the context of mathematical finance.

We note in passing that using machine learning (or deep learning) techniques to solve high-dimensional 
PDEs has recently been the focus of several approaches.
Neural networks have indeed been used to solve PDEs for a long time~\cite{Lagaris1, LeeNN};
more recently, Sirignano and Spiliopoulos~\cite{SpilioDGM} proposed an algorithm
not depending on a given mesh (as opposed to the previous literature),
thus allowing for an easier extension to the multi-dimensional case.
During the writing-up of the present paper, a further two ideas, similar in spirit to the one we are borrowing from~\cite{JentzenDeep} came to light:
Sabate-Vidales, \v{S}i\v{s}ka and Szpruch~\cite{Szpruch}, 
as well as Hur\'e, Pham and Warin~\cite{Pham}
also used the BSDE counterpart of the PDE -- albeit in different ways -- to apply machine learning techniques.
Since our set-up here is not solely about solving a high-dimensional PDE, 
we shall leave the precise comparison of these different schemes to rest for the moment.

The paper will follow a natural progression:
Section~\ref{sec:Modelling} introduces the financial modelling setup of the analysis,
introducing rough local stochastic volatility models, and proving preliminary results fundamental 
for their application in quantitative finance.
In Section~\ref{sec:PPDE}, we show that, in this context, a financial derivative is the solution
to a curve-dependent PDE, for which we propose in Section~\ref{sec:PPDENum} a backward discretisation algorithm, and draw inspiration from the deep learning
methodology developed in~\cite{JentzenDeep}.
Using simulations, we show the validity of this technique in Section~\ref{sec:Numerics} for the rough Heston model, as well as some of its pitfalls.
We gather in an appendix some long proofs and reminders to avoid disrupting the flow of the paper.

\section{Modelling framework}\label{sec:Modelling}
In order to dive right into the modelling framework and our main setup, 
we postpone to Appendix~\ref{sec:Review} a review of the functional It\^o formula
for stochastic Volterra systems, as developed by Viens and Zhang~\cite{Viens}.
We introduce a rough local stochastic volatility model for the dynamics of a stock price process.
Before diving into numerical considerations, we adapt the classical framework of no-arbitrage and market completeness to this setup in order to ensure that pricing and calibration make any sense at all.

\subsection{Rough (local) stochastic volatility model}\label{sec:Model}
We are interested here in stochastic volatility models, where the volatility is rough, 
in the sense of~\cite{Volrough}.
This can be written, under the historical measure, as
\begin{equation}\label{eq:StockSDE}
\left\{\begin{array}{rl}
S_t & = \displaystyle S_0 + \int_{0}^{t}\mu_r S_r \D r + \int_{0}^{t}l(r, S_r, V_r)S_r\D W_r,\\
V_t & = \displaystyle V_0 + \int_{0}^{t}\Kr(t-r)\Big(b(V_r) \D r + \xi(V_r) \D B_r\Big),\\
\D\la W, B\ra_t & = \rho\,\D t,
\end{array}
\right.\end{equation}
where $\rho \in [-1, 0]$, $S_0, V_0$ are strictly positive real numbers, and~$W$ and~$B$ are two standard Brownian motions.
Setting $\Xm = (S, V)$, we can rewrite the system as 
\begin{equation}\label{eq:XSDE}
\Xm_t = \Xm_0 + \int_{0}^{t}\bm(t,r,\Xm_{r})\D r + \int_{0}^{t}\sm(t,r,\Xm_{r})\cdot\D \Bm_r,
\end{equation}
where
\begin{equation}\label{eq:CoeffsSV}
\bm(t, r, \Xm_{r}) = \begin{pmatrix}
\mu_{r} S_{r} \\ \Kr(t-r) b(V_r)
\end{pmatrix}
\qquad\text{and}\qquad
\sm(t, r, \Xm_{r}) = \begin{pmatrix}
\rrho \,l(t, S_{r}, V_{r})S_{r} & \rho \, l(t, S_{r}, V_{r})S_{r}\\
0  & \Kr(t-r) \xi(V_{r})
\end{pmatrix},
\end{equation}
where $\rrho:=\sqrt{1-\rho^2}$
and $\Bm = (B^{\perp}, B)$, with~$W:=\rho B + \rrho B^{\perp}$.
The SDE for~$\Xm$ represents a stochastic Volterra system which is not Markovian in general.
We could in principle allow for more generality and assume, following~\cite{Viens},
that for any $t\geq 0$ and $r \in [0,t]$, the coefficients~$\bm$ and~$\sm$ 
depend on the past trajectory~$\Xm_{r\wedge\cdot}$, for example to include models with delay. 
However, such an extension is not needed in the application we are interested in, and we shall not pursue it.
\begin{remark}\label{rem:LocalVol}
We shall only make assumptions on the behaviour of the function~$l(\cdot)$ that are enough to ensure existence and uniqueness of the system.
A classical example in Mathematical Finance
is $l(t, S, V) = \varsigma(V)$, for some function~$\varsigma(\cdot)$, in which case~\eqref{eq:StockSDE} corresponds to a (rough) stochastic volatility model.
A more general setting is that of a local (rough) stochastic volatility model, with $l(\cdot)$ of the form $l(t, S, v) = \Lr(t, S)\varsigma(v)$, 
where~$\Lr(\cdot, \cdot)$ is called the leverage function.
From the results by Dupire~\cite{DupireVol} and Gy\"ongy~\cite{Gyongy},
if one wants to ensure that this model calibrates exactly to European option prices, 
then the equality
$$
\sigl^2(t, s)
 = \EE^{\QQ}\left[l(t, S_t, V_t)^2\vert S_t = s\right]
  = \EE^{\QQ}\left[\Lr(t, S_t)^2 \varsigma(V_t)^2\vert S_t = s\right]
 = \Lr(t, s)^2\EE^{\QQ}\left[\varsigma(V_t)^2\vert S_t = s\right]
$$
must hold for every $t, s\geq 0$, where the function~$\sigl$ 
is called the local volatility and is obtained (at least in theory) directly from European option prices.
Here~$\QQ$ denotes any given risk-neutral measure (we show later that there are in fact infinitely many of them).
We will show below (Assumption~\ref{assu:Martingale} and Theorem~\ref{thm:NoArb})
that such a probability measure exists, thereby making the model meaningful for option pricing.
The leverage function can then be recovered directly as
\begin{equation}\label{eq:LocalVol}
\Lr(t, s) = \frac{\sigl(t, s)}{\sqrt{\EE^{\QQ}\left[\varsigma(V_t)^2\vert S_t = s\right]}},
\end{equation}
as long as the right-hand side makes sense, and Assumption~\ref{assu:Martingale} below ensures this is indeed the case.
The term on the right-hand side is a conditional expectation with respect to the stock price, and therefore only depends on~$t$ and~$\{S_t=s\}$, no matter whether the variance process is Markovian or not.
This class of models has the advantage of ensuring perfect (at least theoretically) 
calibration to European option prices, while giving flexibility to price other options,
in particular path-dependent or exotic options.
A rigorous proof of the existence and uniqueness of~\eqref{eq:StockSDE} is outside the scope of this paper;
in fact, even in the classical (non-rough case), a general answer does not exist yet, 
and only recent advances~\cite{Jourdain, Lacker} have been made in this direction.
From a numerical perspective, calibration of the leverage function can be performed precisely using the particle
method developed by Guyon and Henry-Labord\`ere~\cite{GuyonPHL}. 
Ideally, and we hope to achieve this in a near future, one should combine the latter with the simulation of the rough volatility component
in order to calibrate vanilla smiles perfectly while capturing other specificities of the market.
\end{remark}
We shall always work under the following considerations:
\begin{assumption}\label{assu:Martingale}\ 
\begin{enumerate}[(i)]
\item The kernel $K\in L^2_{\mathrm{loc}}(\RR_+\to\RR)$ admits a resolvent of the first kind,
and there exists $\gamma\in (0,2]$ such that
$$\int_{0}^{h}\Kr(t)^2\D t = \Oo(h^\gamma)
\quad\text{and}\quad
\int_{0}^{T}\left[\Kr(t+h) - \Kr(t)\right]\D t = \Oo(h^\gamma),
\quad \text{for every }T\geq 0,
\text{ as }h \text{ tends to zero};
$$
\item the functions~$b$ and~$\xi^2$ are linear of the form $b(y) = b_0+b_1 y$ and $\xi(y)^2 = a_0 + a_1 y$;
\item for any $t\geq 0$, 
the map $l(t,\cdot, v)$ is bounded away from zero and bounded above by~$\lr$ for any $v>0$, and 
$l(t,s,\cdot)$ is strictly positive, 
uniformly H\"older continuous 
with $l(t,s,y) \leq C_{\varsigma}\left(1+|y|^{p_{\varsigma}}\right)$
for some $C_{\varsigma}>0$ and $p_{\varsigma} \in (0,1)$;
\item the system~\eqref{eq:StockSDE} admits a unique (weak) solution.
\end{enumerate}
\end{assumption}
The form of the kernel~$\Kr$ ensures that the variance process has stationary increments.
We borrow Condition~(ii) from~\cite{AffineVolterra} so that, in the purely stochastic volatility case
$l(t,s,v)\equiv \sqrt{v}$, we are exactly in the setting of an affine Volterra system 
$(\log(S), V)$.
In fact, this condition alone ensures that the process~$V$ is an affine Volterra process, and by~\cite[Theorem~3.3]{AffineVolterra}, 
Assumption~\ref{assu:Martingale}(i) ensures that the SDE for~$V$ admits a continuous weak solution.
This in particular implies~\cite[Theorem~4.3]{AffineVolterra} that, under suitable integrability conditions,
\begin{equation}\label{eq:Riccati}
\EE\left[\left.\E^{uV_T}\right|\Ff_t\right] = 
\exp\left\{\phi(T-t) + \psi(T-t) V_t\right\},
\end{equation}
for any $0\leq t \leq T$, where the two functions~$\phi$ and~$\psi$ satisfy the system of Riccati equations
$$
\dot{\psi}(t) = b_1\psi(t) + \frac{a_1}{2}\psi(t)^2
\qquad\text{and}\qquad
\dot{\phi}(t) = b_0\psi(t) + \frac{a_0}{2}\psi(t)^2,
$$
with boundary conditions $\psi(0)=u$, $\phi(0)=0$.
Assumption~\ref{assu:Martingale}(i) is again borrowed from~\cite{AffineVolterra}, 
and we refer the interested reader to this paper for examples of kernels satisfying this condition. 
Most examples so far in quantitative finance, such as the power-law kernel 
$\Kr(t) \equiv t^{H-\half}$ and the Gamma kernel
$\Kr(t) \equiv t^{H-\half}\E^{-\lambda t}$, for $H \in (0,1)$, $\lambda>0$, 
fit into this framework.
Finally, Assumption~\ref{assu:Martingale}(iii) allows the representation
\begin{equation}\label{eq:OUExample}
S_t = S_0\exp\left\{\int_{0}^{t}\left(\mu_u - \frac{l^2(u, S_u, V_u)}{2}\right)\D u 
+ \int_{0}^{t}l(u, S_u, V_u)\D W_u\right\}
\end{equation}
to be valid since the stochastic integrand is square integrable, by virtue of
$$
\int_{0}^{t}\EE\left[l^2(u, S_u, V_u)\right]\D u
\leq C_{\varsigma}^2 \lr^2 \int_{0}^{t}\EE\left[\left(1+|V_u|^{2p_{\varsigma}}\right)\right]\D u.
$$
\begin{remark}
The assumption on the function~$l(\cdot)$ may look restrictive, 
in particular regarding the boundedness in the $S$-space;
in the context of local volatility models (Remark~\ref{rem:LocalVol}), this implies boundedness
for each $t\geq 0$ of the map~$\Lr(t, \cdot)$;
when inferring it from market data through~\eqref{eq:LocalVol}, 
it may not be bounded; 
but it is customary to truncate it for large values of the underlying, 
and we follow this convention throughout.
\end{remark}
\begin{remark}
One could consider a slightly different setup as~\eqref{eq:StockSDE}, where the kernel does not apply to the whole dynamics of the process~$V$, but only to the diffusion part, in the spirit of diffusions driven by Volterra Gaussian noises~\cite{BFG15, Volrough, HJL}.
In the simple case, following the seminal paper by Comte and Renault~\cite{Comte-Renault}, 
the variance process is the unique strong solution to the Volterra stochastic equation
\begin{equation}\label{eq:fOU}
V_t = V_0\E^{b t} + a\int_{0}^{t}\E^{b(t-u)}\D B_u^H.
\end{equation}
The fractional Brownian motion with Hurst exponent $H \in (0,1)$
admits the representation $B_t^H = \int_{0}^{t}\Kr(t-u)\D B_u$ for some standard Brownian motion~$B$ 
with the same filtration as~$B^H$~\cite{Ustunel}.
The process~$V$ in~\eqref{eq:fOU} is a continuous Gaussian process with finite variance
and Fernique's estimate~\cite{Fernique} yields that, for any $\alpha, T>0$ and $p<2$,
$\EE\left[\exp\left(\alpha \sup_{t\in [0,T]}|V_t|^p\right)\right]$
is finite, so that
for any $\alpha>0$, 
\begin{align*}
\EE\left[\E^{\alpha\int_{0}^{t}l(u, S_u, V_u)^2\D u}\right]
 & \leq \EE\left[\exp\left\{\alpha C^2_{\varsigma} \lr^2
 \int_{0}^{t}\left(1+|V_u|^{2p_{\varsigma}}\right)\D u \right\}\right]
  \leq  \EE\left[\exp\left\{\alpha C^2_{\varsigma}\lr^2 t
 \left[1 + \sup_{u\in[0,t]}|V_u|^{2p_{\varsigma}}\right] \right\}\right]
\end{align*}
is finite and~\eqref{eq:OUExample} holds.
\end{remark}

\subsection{Market completeness and arbitrage freeness}
In the general case where the correlation parameter is different from~$-1$ and~$1$, the presence of the two noises renders the market incomplete.
In order to be able to use the system~\eqref{eq:StockSDE} for pricing purposes, 
we need to show how to complete the market, 
and to check for the existence of some probability measure under which the stock price is a true martingale.
The latter issue was solved for the rough Heston model by El Euch and Rosenbaum~\cite{ERHedging},
while the rough Bergomi case with non-positive correlation was recently proved by Gassiat~\cite{Gassiat}. 
We show that this still holds in our framework.
We assume the existence of a money market account yielding a risk-free 
interest rate~$(r_t)_{t\geq 0}$.
\begin{theorem}\label{thm:NoArb}
The market is incomplete and free of arbitrage.
\end{theorem}
\begin{proof}
We wrote~\eqref{eq:StockSDE} under the historical measure~$\PP$, but
for pricing purposes, we need it under the pricing measure~$\QQ$.
Introduce the market prices of risk as adapted processes~$(\lambda_t)_{t\geq 0}$ and~$(\beta_t)_{t\geq 0}$ such that
\begin{equation}\label{eq:ChangeMeasure}
\rho \lambda_t + \rrho\beta_t = \frac{\mu_t - r_t}{l(t, S_t, V_t)},
\end{equation}
which is well defined since the denominator is strictly positive by Assumption~\ref{assu:Martingale}(iii).
We now define the probability measure~$\QQ$ via its Radon-Nikodym derivative
$$
\left.\frac{\D\QQ}{\D\PP}\right|_{\Ff_t} := \exp\left\{-\frac{1}{2}\int_{0}^{t}\left(\lambda_s^2 + \beta_s^2\right)\D s - \int_{0}^{t}\left(\lambda_s\D B_s + \beta_s \D B_s^\perp\right)\right\},
$$
so that Girsanov's Theorem~\cite[Chapter 3.5]{Karatzas} implies that the processes $B^{\QQ}$ and $B^{\QQ, \perp}$ defined by 
$$
B^{\QQ}_t := B_t + \int_{0}^{t}\lambda_u \D u
\qquad\text{and}\qquad
B^{\QQ, \perp}_t := B_t^{\perp} + \int_{0}^{t}\beta_u\D u
$$
are orthogonal Brownian motions under~$\QQ$.
Therefore, under~$\QQ$, the stock price satisfies
\begin{equation}\label{eq:StockSDEQ}
\begin{array}{rl}
S_t & = \displaystyle S_0 + \int_{0}^{t}r_u S_u\D u + \int_{0}^{t}l(u, S_u, V_u)
\left(\rho\,\D B_u^{\QQ} + \rrho\,\D B_u^{\QQ,\perp}\right),\\
V_t & = \displaystyle V_0 + \int_{0}^{t}\Kr(t-u)\left(\widehat{b}(V_u)\D u + \xi(V_u) \D B_u^{\QQ}\right),
\end{array}
\end{equation}
where, under~$\QQ$, the drift of the variance process is now of the form 
$\widehat{b}(V_t) = b(V_t) - \lambda_t\xi(V_t)$.
Introducing the Brownian motion~$W^{\QQ}$ (under~$\QQ$), the first SDE then reads
$$
S_t = S_0 + \int_{0}^{t}r_u S_u\D u + \int_{0}^{t}l(u, S_u, V_u) \D W^{\QQ}.
$$
The proposition then follows directly from Novikov's criterion using 
Assumption~\ref{assu:Martingale}(iii).
\end{proof}

\begin{remark}\label{rem:rBergomi}
We can relax the assumption on~$l(\cdot)$.
Assume for example a rough local stochastic volatility setting (Remark~\ref{rem:LocalVol}) where
$l(t, S, v) = \Lr(t, S)\varsigma(v)$, but with Assumption~\ref{assu:Martingale}(iii) replaced by
$\varsigma(v) = \E^{\eta v}$, for $\eta>0$.
Since $\rho\leq 0$, the proof of~\cite[Theorem 1]{Gassiat}, 
using a localisation argument with the increasing sequence of stopping times
$\tau_n:=\inf\{t>0: V_t=n\}$, remains the same, and the stock price is a true martingale.
Thus a local volatility version of the rough Bergomi model~\cite{BFG15, JMM17, JPS17}
falls within our no-arbitrage setting.
\end{remark}

\section{Pricing via curve-dependent PDEs}\label{sec:PPDE}
We shall from now on only work under the risk-neutral measure and, 
with a slight abuse of notations, write~$b$ instead of~$\widehat{b}$
for the drift of the variance process (equivalently taking the market price of volatility risk to be null
in Theorem~\ref{thm:NoArb}).
In the classical It\^o diffusion setting, where the kernel~$\Kr$ is constant, 
the Feynman-Kac formula transforms the pricing problem from a probabilistic setting 
to a PDE formulation. 
The formulation of our rough local stochastic volatility model~\eqref{eq:StockSDE}
goes beyond this scope since the system is not Markovian any longer.
We adapt here the methodology developed by Viens and Zhang~\cite{Viens}
to show that the pricing problem is equivalent to solving a curve-dependent PDE.

We first start with the following lemma, which shows, not surprisingly, that the option price
should be viewed, not as a function of the state variable at a fixed given time, but 
as a functional over paths.
The argument is not new, and versions have already appeared for the rough Heston model in~\cite{ERHedging}
and in the context of It\^o's formula for stochastic Volterra equations in~\cite{Viens}.
Consider a European option with maturity~$T$ and payoff function~$g(\cdot)$ written on the stock price~$S$.
By standard no-arbitrage arguments, from Theorem~\ref{thm:NoArb}, 
its price at time~$t\in [0,T]$ reads 
\begin{equation}\label{eq:Price}
P_t = \EE[g(S_T)\vert\Ff_t].
\end{equation}
The key object in our analysis below is an infinite-dimensional stochastic process 
$(\Theta^t)_{t\in [0,T]}$, adapted to the filtration $(\Ff_t)_{t\in[0,T]}$, such that the following assumption holds,
making the price~\eqref{eq:Price} at time $t\leq T$ as a functional of $\Ff_t$-measurable quantities:

\begin{assumption}\label{assu:Functional}
There exist a process $(\Theta^t)_{t\in [0,T]}$, adapted to $(\Ff_{t})_{t\in [0,T]}$, and a map $\Ps:[0, T]\times[0,\infty)\times \Cc([0,T],\RR)\to\RR$, 
such that, for any $t\in [0,T]$, 
$$
P_t = \Ps\left(t, S_t, \left(\Theta^t_u\right)_{t\leq u\leq T}\right).
$$
\end{assumption}

For any $t\in[0,T]$, the one-dimensional curve~$\Theta^t$ 
(or equivalently the infinite-dimensional process $(\Theta^t)_{t\in [0,T]}$) 
plays a fundamental role and is related, though different, to the forward variance curve $u\mapsto \EE[V_u|\Ff_t]$.
It is far from trivial to show such a representation for general rough local stochastic volatility models.

\begin{example}[Rough Heston]
In the rough Heston model, the stock price satisfies the SDE~\eqref{eq:StockSDE} with
$b(v) = \kappa(\theta-v)$ and $\xi(v) = \nu\sqrt{v}$ for $\kappa, \theta,\nu>0$ (see also~\eqref{eq:RoughHeston} below).
In this case, defining
\begin{equation}\label{eq:Theta}
\Theta_{u}^{t} := 
\left\{
\begin{array}{ll}
V_u, & \text{if } u \in [0,t],\\
\displaystyle \EE\left[\left. V_u - \int_{t}^{u}b(u,r,V_{r})\D r\right|\Ff_t\right], & \text{if } u \in [t, T],
\end{array}
\right.
\end{equation}
then Assumption~\ref{assu:Functional} holds as detailed in~\cite[Section~5]{Viens}. 
Here~$\Theta^t$ only corresponds to the forward variance process $\EE[V_u|\Ff_t]$ of~\cite{ERHedging} only when the variance drift~$b$ is null.
\end{example}

\begin{example}[Rough Bergomi]
It also holds for the rough Bergomi model~\cite{BFG15} with a little more work, as explained in~\cite[Section~5]{Viens}.
In this model, the dynamics read
\begin{equation}\label{eq:rBergomi}
\left\{\begin{array}{rl}
S_t & = \displaystyle S_0 + \int_{0}^{t}\mu_r S_r \D r + \int_{0}^{t}\sqrt{V_t}S_r\D W_r,\\
V_t & = \displaystyle  V_0 \exp\left\{M_t - \half \lambda^2 t^2\right\}.
\end{array}
\right.\end{equation}
with $M_t = \lambda\sqrt{2H}\int_{0}^{t}(t-r)^{H-\half}\D B_r$.
We can rewrite~\eqref{eq:rBergomi} as a $(S,M)$ dynamics as
\begin{equation*}
\left\{\begin{array}{rl}
S_t & = \displaystyle S_0 + \int_{0}^{t}\mu_r S_r \D r + \int_{0}^{t}\sqrt{V_0} \exp\left\{\half M_r - \frac{1}{4}\lambda^2 r^2\right\}S_r\D W_r,\\
M_t & = \displaystyle  \lambda\sqrt{2H}\int_{0}^{t}(t-r)^{H-\half}\D B_r.
\end{array}
\right.
\end{equation*}
The couple $\Xm = (S, M)$ can therefore be rewritten as in~\eqref{eq:StockSDE}.
In light of Remark~\ref{rem:rBergomi}, the exponential function in the diffusion part of~$S$ violates Assumption~\ref{assu:Martingale}(iii) per se,
but the process remains a martingale (at least for $\rho\leq 0$ by~\cite{Gassiat}), which is enough for our purposes.
Setting 
\begin{equation*}
\Theta^t_u := 
\left\{
\begin{array}{ll}
M_u, & \text{if } u \in [0,t],\\
\displaystyle \EE[M_u|\Ff_t] = \lambda\sqrt{2H}\int_{0}^{t}(u-r)^{H-\half}\D B_r, & \text{if } u \in [t, T],
\end{array}
\right.
\end{equation*}
Assumption~\ref{assu:Functional} then holds~\cite[Section~5]{Viens}.
\end{example}
The following theorem is the main result here (proved in Appendix~\ref{sec:ProofThmPDEBSDE}), 
and shows how to extend the classical Feynman-Kac formula to the curve-dependent case. 
From now on, we shall adopt the notation
\begin{equation}\label{eq:K}
\Kr^t:= \Kr(\cdot-t), \qquad\text{for any }t\geq 0,
\end{equation}
to denote the curve~$\Kr^t$ seen at time~$t$. It will become clear in the following theorem how this becomes handy.
\begin{theorem}\label{thm:PDEBSDE}
Under Assumption~\ref{assu:Functional}, 
the option price~\eqref{eq:Price} is the unique solution to the linear curve-dependent PDE 
\begin{equation}\label{eq:CPDE}
\Big(\dt + \Ll_{x} + \Ll_{xx} + \Ll_{x\omega} + \Ll_{\omega} + \Ll_{\omega\omega} - r_t\Big)\Ps\left(t, S_t, \Theta^t\right) = 0,
\end{equation}
for $t\in [0, T)$, with boundary condition $P\left(T, S_T, \Theta^T\right) = g(S_T)$, 
where, at the point $(t, x, \Theta^t)$,
\begin{equation*}
\begin{array}{rlrl}
\Ll_{x\omega} & := \displaystyle\rho l(t,x,\Theta^t_t)\xi(\Theta^t_t) x\la
\partial_{x,\omega}, \Kr^t\ra,
& \Ll_{xx} & := \displaystyle\frac{1}{2} l(t,x,\Theta^t_t)^2x^2\dxx,
\qquad \Ll_{x} := r_t x \dx,\\
\Ll_{\omega\omega} & := \displaystyle\frac{1}{2}\xi(\Theta^t_t)^2\la\domom, (\Kr^t, \Kr^t)\ra,
& \Ll_{\omega} & := \displaystyle b(\Theta^t_t)\la\dom, \Kr^t\ra.
\end{array}
\end{equation*}
\end{theorem}
The pricing PDE~\eqref{eq:CPDE} has both state-dependent terms, involving derivatives with respect to~$x$,
and curve-dependent ones, involving functional derivatives with respect to~$\omega$.
In order to streamline the presentation, we defer to Appendix~\ref{sec:Review}
the precise framework, borrowed from~\cite{Viens}, to define these derivatives.
Without essential loss of understanding for the rest of the analysis, 
the reader can view them basically as Fr\'echet derivatives 
(with some regularisation due to the singularity of the kernel along the diagonal).
We would like to point out, however that our framework, supported by Assumption~\ref{assu:Functional},
concerns functions that depend on points~$S_t$ and on curves~$\Theta^t$ rather than on full history-dependent paths, 
as is done for example in~\cite{Touzi3, Touzi2, Touzi1}.
During the last revision stages of the present work, Bayer, Qiu and Yao~\cite{BayerSPDE} 
extended this approach to more general rough stochastic volatility models, via the use of backward stochastic PDEs,
investigating the existence of such equations, and developing a deep-learning based algorithm to solve them.

\begin{remark}
Assumption~\ref{assu:Functional} yields Theorem~\ref{thm:PDEBSDE} in the sense that, when evaluating the European option at time $t \in [0,T)$,
the price only depends on the value of the underlying~$S$ at time~$t$ and on the curve~$(\Theta^t_u)_{u \in [t,T]}$.
It is therefore not a fully path-dependent PDE, explaining the terminology `curve-dependent' rather than `path-dependent'.
While one could appeal directly to techniques tailored for curve-dependent PDEs, 
stating the results in the more general framework of path-dependent PDEs 
allow for greater generality and extensions to the fully path-dependent case.
One such example is that of a floating leg of a variance swap.
In this case, we can write the price at time $t\in [0,T]$ as
$$
P_t = \EE\left[\frac{1}{T}\int_{0}^{T}V_u\D u\vert\Ff_t\right]
 = \frac{1}{T}\EE\left[\int_{0}^{t}V_u\D u\vert\Ff_t\right] + \frac{1}{T}\EE\left[\int_{t}^{T}V_u\D u\vert\Ff_t\right]
 = \frac{1}{T}\int_{0}^{t}V_u\D u + \frac{1}{T}\int_{t}^{T}\EE[V_u\vert\Ff_t]\D u.
$$
Writing
\begin{equation}\label{eq:ThetaVarSwap}
\Theta_{u}^{t} := 
\left\{
\begin{array}{ll}
V_u, & \text{if } u \in [0,t],\\
\displaystyle \EE[V_u \vert\Ff_t], & \text{if } u \in [t, T],
\end{array}
\right.
\end{equation}
we see that, at time~$t$, the price $P_t$ is now a function of the whole path~$(\Theta^t_u)_{u \in [0,T]}$
and not just of the (forward) curve $(\Theta^t_u)_{u \in [t,T]}$.
\end{remark}

\begin{remark}
Another example covered in this framework is the following: consider an option with payoff
$$
P_T = g(S_T) + \int_{0}^{T}f(S_{u})\D u,
$$
for some functions~$f$ and~$g$ with sufficient smoothness and growth conditions.
At time $t\in [0,T]$, no-arbitrage arguments imply that, under the risk-neutral measure,
\begin{align*}
P_t = \EE[P_T\vert \Ff_t]
 & = \EE\left[g(S_T) + \int_{0}^{T}f(S_{u})\D u \Big\vert \Ff_t\right]\\
 & = \int_{0}^{t}f(u,S_{u})\D u + \EE\left[g(S_T)\vert \Ff_t\right] + \EE\left[\int_{t}^{T}f(S_{u})\D u\Big\vert \Ff_t\right]\\
 & = \int_{0}^{t}f(u,S_{u})\D u + \Ps^g\left(t, S_t, \left(\Theta^t_r\right)_{t\leq r\leq T}\right)
  + \int_{t}^{T}\Ps^f\left(t, S_t, \left(\Theta^t_r\right)_{t\leq r\leq u}\right),
\end{align*}
where~$\Ps^g(t, \cdots)$ corresponds to the price at time~$t$ of a European option with payoff~$g(\cdot)$ and maturity~$T$,
while~$\Ps^f(t, \cdots)$ to the price of an option at time~$t$  with payoff~$f(\cdot)$ and maturity~$u$.
In that case, our framework can accommodate this separately for the $g$-option and the $f$-option (with some extra cost for the discretisation of the integral).
\end{remark}

\section{Numerical framework for CPDEs}\label{sec:PPDENum}
Theorem~\ref{thm:PDEBSDE} showed that pricing under rough volatility could be analysed through 
the lens of path-dependent (or curve-dependent here
 PDEs.
However, numerical schemes for such equations are scarce, 
and the only approaches we are aware of is the extension of 
Barles and Souganidis' monotone scheme~\cite{Barles} to the path-dependent case by Zhang and Zhuo~\cite{Zhang},
the convergence of which was proved by Ren and Tan~\cite{RenTan}.
However, the actual implementation of this scheme in the PPDE context is far from trivial,
and we consider a different route here, more amenable to computations in our opinion, at least in our curve-dependent framework.
We first discretise the CPDE along some basis of functions, 
reducing the infinite-dimensional problem to a finite-, yet high-, dimensional problem.
High-dimensional PDEs suffer from the so-called curse of dimensionality,
and are notoriously difficult to solve.
We then develop a backward algorithm inspired by the method in~\cite{JentzenDeep}
to solve this system of PDEs.

\subsection{Discretisation of the CPDE}
For each $t\in [0,T]$, we consider a basis $\psib^t = (\psi_a^t)_{a=1,\ldots, p}$  of c\`adl\`ag functions,
for some fixed integer~$p$,
and use it to approximate~$\Theta^t$ and~$\Kr^t$ by
$$
\widehat{\Theta}^t := \thetab^t \cdot(\psib^t)^\top
\qquad\text{and}\qquad
\widehat{\Kr}^t := \kkappa^t \cdot(\psib^t)^\top,
$$
for some sequence of real coefficients 
$\thetab^t := (\theta^t_{a})_{a=1,\ldots, p}$ and $\kkappa^t:=(\kappa^t_{a})_{a=1,\ldots, p}$.
Since $\Kr^t\in\Dd_t$,
the space of c\`adl\`ag functions on $[t, T]$ (see Appendix~\ref{sec:FuncIto}), then, 
from Definition~\ref{def:FuncDeriv},
$$
\la\dom \Ps\left(t, x, \Theta^t\right), \Kr^t\ra 
:= \partial_{\eps} \Ps\left.\left(t, x, \Theta^t + \eps \Kr^t\,\ind_{[t,T]}\right)\right|_{\eps=0}
= \partial_{\eps} \Ps\left.\left(t, x, \Theta^t + \eps \Kr^t\right)\right|_{\eps=0},
$$
and we can introduce the following approximations of the path derivatives along the direction~$\Kr^t$:
\begin{align*}
\la\dom \Ps\left(t, x, \widehat{\Theta}^t\right), \widehat{\Kr}^t\ra
 & := \partial_{\eps} \Ps\left.\left(t, x, \widehat{\Theta}^t + \eps \widehat{\Kr}^t\right)\right|_{\eps=0}
 = \partial_{\eps} \Ps\left.\left(t, x,\sum_{a=1}^{p}\Big(\theta^t_a + \eps \kappa^t_a\Big)\psi_a\right)\right|_{\eps=0}\\
 & =: \partial_{\eps} \PHat\left.\left(t, x,\left(\theta^t_a + \eps \kappa^t_a\right)_{a=1}^{p}\right)\right|_{\eps=0}
 = \sum_{a=1}^{p}\partial_{\theta^t_a} \PHat\left(t, x,\thetab^t\right)\kappa^t_a
  = \nabla_{\thetab^t}\PHat\left(t, x,\thetab^t\right)\cdot \kkappa^t,
\end{align*}
where the new function~$\PHat$ now acts on $[0,T]\times[0,\infty)\times\RR^p$.
Likewise, for the second functional derivative, 
$$
\la\domom \Ps\left(t, x,  \widehat{\Theta}^t\right), \left(\widehat{\Kr}^t,  \widehat{\Kr}^t\right)\ra
  = \sum_{a,j=1}^{p}\partial_{\theta^t_a \theta^t_j} \PHat\left(t, x,\thetab^t\right) \kappa^{t}_a \kappa^{t}_j
  = \left(\kkappa^t\right)^\top \cdot \Delta_{\thetab^t} \PHat\left(t, x,\thetab^t\right)\cdot\kkappa^t,
$$
and finally the cross derivatives can be approximated similarly as
$$
\la\dxom \Ps\left(t, x,  \widehat{\Theta}^t\right), \widehat{\Kr}^t\ra
 := \partial_{x}\nabla_{\thetab^t}\PHat\left(t, x,\thetab^t\right)\cdot\kkappa^t.
$$
The CPDE~\eqref{eq:CPDE} therefore becomes
\begin{equation}\label{eq:CPDEDiscrete}
\left(\dt + \Ll_{x} + \Ll_{xx} + \sum_{a=1}^{p}\Ll_{x\theta_a^t} + \sum_{a=1}^{p}\Ll_{\theta_a^t}
 + \sum_{a,j =1}^{p}\Ll_{\theta_a^t \theta_j^t} - r_t\right)\PHat = 0,
\end{equation}
where the differential operators are defined, for each $a,j=1, \ldots, p$, as 
\begin{equation*}
\begin{array}{rlrl}
\Ll_{x\theta_a} & := \displaystyle \rho\, l(t,x, \thetab^t) \xi(\thetab^t) x \kappa^t_a \partial_{x \theta_a^t},
& \Ll_{xx} & := \displaystyle\frac{l(t,x, \thetab^t)^2}{2} x^2 \dxx,
\qquad \Ll_{x} := r_t x \dx,\\
\Ll_{\theta_a^t \theta_j^t} & := \displaystyle \frac{\xi(\thetab^t)^2}{2}\kappa^t_a \kappa^t_j \partial_{\theta_a^t \theta_j^t},
& \Ll_{\theta_a} & := \displaystyle b(\thetab^t)\kappa^t_a \partial_{\theta_a^t}.
\end{array}
\end{equation*}
We can rewrite this system in a more concise way as 
\begin{equation}\label{eq:PDEPHat}
\dt \PHat + \frac{1}{2}\mathrm{Tr}\left(\Sigb\cdot\Sigb^{\top}\cdot\Delta\PHat\right) + {\boldsymbol\mu}\cdot\nabla \PHat - r_t\PHat = 0,
\end{equation}
where
${\boldsymbol\mu}\left(t, x, \thetab^t\right) := \left(x r_t, b(\thetab^t)\kappa^t_1, \ldots, b(\thetab^t)\kappa^t_{p}\right)^{\top}$
and
$$
\Sigb\left(t, x, \thetab^t\right) \cdot \Sigb\left(t, x, \thetab^t\right)^\top := 
\begin{pmatrix}
l(t, x, \thetab^t)^2 x^2 & \rho l(t, x, \thetab^t)\xi(\thetab^t) x \kappa^t_1
 & \cdots & \rho l(t, x, \thetab^t)\xi(\thetab^t) x \kappa^t_{p}\\
\rho\, l(t,x, \thetab^t) \xi(\thetab^t) x \kappa^t_1 & \xi(\thetab^t)^2 (\kappa^t_1)^2 & \cdots & \xi(\thetab^t)^2 \kappa^t_1 \kappa^t_{p}\\
\vdots & \vdots & \ddots & \vdots \\
 \rho\, l(t,x, \thetab^t) \xi(\thetab^t) x \kappa^t_{p} & \xi(\thetab^t)^2 \kappa^t_{1} \kappa^t_{p}
  & \cdots
 & \xi(\thetab^t)^2 \left(\kappa^t_{p}\right)^2
\end{pmatrix}
$$

\begin{remark}\label{rem:Simple}
The simplest example is to consider piecewise constant curves
$\psi^t_a = \ind_{\delta^t_a}$,
for $a=1,\ldots, p$, where $(\delta^t_a)_{a=1}^{p}$ represents a mesh of the interval~$[t, T]$.
To simplify the notations below, we shall consider the mesh
$\delta^t_a = [t_{a-1}, t_a)$, and we write
$\theta_a^t = \Theta^t_{t_a}$ and $\kappa^t_a = \Kr(t_a - t)$.
\end{remark}


There is an interesting connection between the functional It\^o formula in Theorem~\ref{thm:Ito} and backward stochastic differential equations.
Following~\cite{KarouiBSDE}, consider the multidimensional BSDE
\begin{equation}\label{eq:BSDE}
\left\{
\begin{array}{rl}
X_t = & \displaystyle \xi + \int_{0}^{t}\overline{\mu}(r, X_r)\D r + \int_{0}^{t}\overline{\Sigma}(r, X_r)\D W_r,\\
Y_t = & \displaystyle g(X_T) + \int_{t}^{T}f(r, X_r, Y_r, Z_r)\D r - \int_{t}^{T}Z_r^\top \cdot \D W_r,
\end{array}
\right.
\end{equation}
on some filtered probability space $(\Omega, \Ff, (\Ff_t)_{t\geq 0}, \PP)$ supporting a $d$-dimensional Brownian motion~$W$.
The solution process $(X, Y, Z)$ takes values, 
at each point in time, in $\RR^d\times\RR\times\RR^d$.
Consider further the PDE
\begin{equation}\label{eq:PDE}
\dt u(t, x) +\frac{1}{2}\mathrm{Tr}\Big(\overline{\Sigma}(t,x)\overline{\Sigma}(t,x)^\top \Delta u(t, x)\Big)
 + \overline{\mu}(t,x) \nabla u(t, x) + f\Big(t, x, u(t,x), \overline{\Sigma}(t,x)^\top\nabla u(t,x)\Big) = 0,
\end{equation}
with terminal boundary condition $u(T,x) = g(x)$.
In the regular case (Assumption~\ref{assu:Kernel}(i)), as shown in~\cite{Viens}, 
if~\eqref{eq:PDE} has a classical solution in $\Cc^{1,2}_+(\Lambda)$ (see Appendix~\ref{sec:FuncIto}), then the couple $(Y,Z)$
defined as $Y_t := u(t,X_t)$ and $Z_t := \overline{\Sigma}(t,X_t)^\top \cdot \nabla u(t, X_t)$ 
is the solution to~\eqref{eq:BSDE}.
In light of this result, and getting inspiration from~\cite{JentzenDeep},
if the solution to the option problem satisfies~\eqref{eq:PDEPHat}, it also solves the BSDE~\eqref{eq:BSDE},
namely
\begin{equation}\label{eq:BSDEPriceBwd}
\PHat\left(t, S_t, \thetab^t\right)
 = \PHat\left(T, S_T, \thetab^T\right)
 + \int_{t}^{T} r_u\PHat\left(u, S_u, \thetab^u\right)\D u
 - \int_{t}^{T} \Sigb\left(u, S_u, \thetab^u\right)^\top \cdot \nabla\PHat\left(u, S_u, \thetab^u\right)\D W_u,
\end{equation}
or, written in forward form, 
\begin{equation}\label{eq:BSDEPriceFwd}
\PHat\left(t, S_t, \thetab^t\right)
 = \PHat\left(0, S_0, \thetab^0\right)
 - \int_{0}^{t} r_u\PHat\left(u, S_u, \thetab^u\right)\D u
 + \int_{0}^{t} \Sigb\left(u, S_u, \thetab^u\right)^\top \cdot \nabla\PHat\left(u, S_u, \thetab^u\right)\D W_u.
\end{equation}

\subsection{Neural network structure}
We now introduce the neural network structure that will help us solve the high-dimensional pricing problem above.
We concentrate on the setting in Remark~\ref{rem:Simple}.

\subsubsection{Simulation of the network inputs}
We first discretise in time the stochastic Volterra system~\eqref{eq:XSDE}.
Many different possible discretisation schemes exist, and we shall not here explore them in great details.
The fundamental feature here is the singularity of the kernel on the diagonal, which requires special care
and is dealt with using a hybrid scheme, recently developed by Bennedsen, Lunde and Pakkanen~\cite{BLP15}.
We postpone to Appendix~\ref{appsec:rHestonSimul} a detailed analysis of the discretisation for the rough Heston model, 
which we will use in our numerical application later on.
For the sake of our argument here, all we require at the moment is a discretised process
$(S_{t_i}, (\Theta_{t_j}^{t_i})_{0\leq j\leq n})$ along a grid $(t_i = \frac{iT}{n})_{0\leq i\leq n}$, for some integer~$n$.

\subsubsection{Euler discretisation scheme}
We iteratively compute the price of the option at each time step.
Both~$\PHat_0$ and~$(\nabla\PHat_0)$ are the initial price and gradient that will be optimised with the weights of the network.
On the two-dimensional grid $(t_i,t_j)_{0\leq i,j\leq n}$, 
we can discretise the forward stochastic equation~\eqref{eq:BSDEPriceFwd} as
\begin{equation*}
\left\{
\begin{array}{rcl}
\PHat\left(t_0, S_{t_0}, \thetab^{t_0}\right) & = & \PHat_0,\\
\nabla \PHat\left(t_0, S_{t_0}, \thetab^{t_0}\right) & = & (\nabla\PHat)_0,\\
\PHat\left(t_{i+1}, S_{t_{i+1}}, \thetab^{t_{i+1}}\right) & = & 
\left(1-r_{t_i}\Delta_{i}\right) \PHat\left(t_{i}, S_{t_{i}}, \thetab^{t_{i}}\right)
 + \Sigma\left(t_i, S_{t_i}, \thetab^{t_i}\right)^\top \nabla\PHat\left(t_{i}, S_{t_{i}}, \thetab^{t_{i}}\right)\Delta W_{t_i},
\end{array}
\right.
\end{equation*}
where $W_{t_i} = (B^{\QQ,\perp}_{t_i}, B_{t_i},\ldots, B_{t_i})^{\top}\in\RR^{p+1}$.
Discretising the backward SDE~\eqref{eq:BSDEPriceBwd} would give
\begin{equation}\label{eq:BwdScheme}
\left\{
\begin{array}{rl}
\PHat\left(t_n, S_{t_n}, \thetab^{t_n}\right) & = g(S_T),\\
\PHat\left(t_{i}, S_{t_{i}}, \thetab^{t_{i}}\right) & = 
\left(1+r_{t_{i+1}}\Delta_{i+1}\right) \PHat\left(t_{i+1}, S_{t_{i+1}}, \thetab^{t_{i+1}}\right)
 - \Sigma\left(t_i, S_{t_i}, \thetab^{t_i}\right)^\top \nabla\PHat\left(t_{i}, S_{t_{i}}, \thetab^{t_{i}}\right)\Delta W_{t_i}.
\end{array}
\right.
\end{equation}
We follow here this backward approach, more natural for pricing exotic (path-dependent) derivatives,
such as Bermudan or American options.
This is, strictly speaking, a Forward-Backward approach as we simulate the stock price forward,
and then the option price backward, but we stick to the `Backward' terminology.
\begin{remark}
Standard Euler schemes to discretise the backward SDE~\eqref{eq:BSDE} 
are obviously faced with measurability issues of the~$Z_{t_i} =  \Sigma(t_i, \cdots)^\top\nabla\PHat(t_i, \cdots)$ component,
and we refer the reader to~\cite{Bouchard, Peng} for different related schemes, 
all essentially based on taking conditional expectations at time~$t_i$.
Here, however, we are not computing this term exactly but, as detailed below,
are using a neural network to learn it.
So in fact, the network learns the measurable version of this term along the discretised time grid.
This does not guarantee measurability of~$\PHat(t_i, \cdots)$ on the left-hand side of~\eqref{eq:BwdScheme} though. Refinements of this scheme to tackle this issue will be investigated in future works.
\end{remark}
\begin{remark}
It would be interesting to compare the discretisation and the neural network approach for both the forward and the backward problems.
However, in order to keep the focus on the paper on our ultimate goal 
(a numerical scheme for rough volatility models), 
we leave this suggestion to further research.
\end{remark}

\subsection{Neural networks}\label{sec:SubNetwork}
Based on the discretisation of the process, we introduce a coarser discretisation grid~$(\tau_i)_{i=1,\ldots, m}$
such that $\tau_m=t_n$, $\tau_0=t_0$ and $m\leq n$.
In~\cite{JentzenDeep}, E, Han and Jentzen assumed $n=m$, but we allow here for more flexibility.
This also greatly improves the speed of the algorithm, as it reduces the number of networks--hence the number of network parameters--and 
simplifies the computation of the loss function.
For each $i=1, \ldots, m$, the price as well as the model parameters are known.
The only unknown is the term $\Sigma\left(\tau_i, S_{\tau_i}, \thetab^{\tau_i}\right)^\top \nabla\PHat\left(\tau_{i}, S_{\tau_{i}}, \thetab^{\tau_{i}}\right)$ involving the gradient and the diffusion matrix. 
We therefore use the neural network below to infer its value.
The input consists of one input layer with the value of the processes at~$t_i$.
Each hidden layer is computed by multiplying the previous layer by the weights~$\wwb$ 
and adding a bias~$\ddelta$.
After computing a layer, we apply a batch normalisation by computing the mean~$\mathfrak{m}$ and the standard deviation~$s$ of the layer, 
and by applying the linear transformation $\mathrm{T}(x) := \gamma\frac{x-\mathfrak{m}}{s} + \beta$ to each element of the layer,
where the scale~$\gamma$ and the offset~$\beta$ are to be calibrated.
We also use the \texttt{ReLu} activation function $a(x) = x_+$ on each element of the layer.

\tikzset{%
  every neuron/.style={
    rectangle,
    draw,
    minimum size=1cm,
    rounded corners=8
  },
    noDraw/.style={
    rectangle,
    minimum size=1cm,
    rounded corners=8
  },
  neuron missing/.style={
    draw=none, 
    scale=4,
    text height=0.333cm,
    execute at begin node=\color{red}$\vdots$
  },
}

\begin{tikzpicture}[x=1.5cm, y=1.5cm, >=stealth]
\node [every neuron/.try, neuron /.try] (-) at (-1.5,-2.5) {$\PHat\left(\tau_{i-1}, S_{\tau_{i-1}}, \Th^{\tau_{i-1}}\right)$};
\node [every neuron/.try, neuron /.try] (-) at (-1.5,-7.5) {$S_{\tau_{i-1}}, \Th^{\tau_{i-1}}$};

\draw [<-] (-0.3,-2.5)-- ++(0.3,0)    node [above, midway] {};
\draw [<-] (1.9,-2.5)-- ++(0.4,0)    node [above, midway] {};
\draw [<-] (4.6,-2.5)-- ++(0.2,0)    node [above, midway] {};
\draw [<-] (5.25,-2.5)-- ++(0.25,0)    node [above, midway] {};
\draw [->] (1.3,-3.1)-- ++(-0.,0.2)    node [above, midway] {}; 
\draw [->] (-0.4,-7.5)-- ++(0.7,0)    node [above, midway] {};
\draw [decorate,decoration={brace,amplitude=10pt,raise=4pt},yshift=0pt]
(4,-4.2) -- (4,-6.7) node [black,midway,xshift=2.3cm] {\footnotesize Multilayer neural network};
\draw [decorate,decoration={brace,amplitude=5pt,mirror,raise=4pt},yshift=0pt]
(-0.4,-3.2) -- (-0.4,-3.8) node [black,midway,xshift=-1.3cm] {\footnotesize Output data};
\draw [decorate,decoration={brace,amplitude=5pt,raise=4pt},yshift=0pt]
(4.2,-7.2) -- (4.2,-7.8) node [black,midway,xshift=1.3cm] {\footnotesize Input data};

\node [every neuron/.try, neuron /.try] (-) at (1,-2.5) {$\PHat\left(\tau_{i}, S_{\tau_{i}}, \Th^{\tau_{i}}\right)$};
\node [every neuron/.try, neuron /.try] (-) at (1,-3.5) {$(\Sigma^\top\cdot\nabla\PHat)\left(\tau_{i}, S_{\tau_{i}}, \Th^{\tau_{i}}\right)$};
\node [every neuron/.try, neuron /.try] (-) at (1,-4.5) {Layer L};
\node [noDraw/.try, neuron /.try] (-) at (1,-5.5) {$\vdots$};
\node [every neuron/.try, neuron /.try] (-) at (1,-6.5) {Layer $1$};
\node [every neuron/.try, neuron /.try] (-) at (1,-7.5) {$S_{\tau_{i}}, \Th^{\tau_{i}}$};
\node [every neuron/.try, neuron /.try] (-) at (1,-8.5) {$B_{\tau_{i}}-B_{\tau_{i-1}}, B_{\tau_{i}}^{\perp} - B_{\tau_{i-1}}^{\perp}$};

\draw [->] (1,-8.1)-- ++(0,0.2)    node [above, midway] {};
\draw [->] (1,-7.1)-- ++(0,0.2)    node [above, midway] {};
\draw [->] (1,-6.1)-- ++(0,0.2)    node [above, midway] {};
\draw [->] (1,-5.3)-- ++(0,0.4)    node [above, midway] {};
\draw [->] (1,-4.1)-- ++(0,0.2)    node [above, midway] {};

\node [every neuron/.try, neuron /.try] (-) at (3.5,-2.5) {$\PHat\left(\tau_{i+1}, S_{\tau_{i+1}}, \Th^{\tau_{i+1}}\right)$};
\node [noDraw/.try, neuron /.try] (-) at (5.1,-2.5) {$\cdots$};

\node [every neuron/.try, neuron /.try] (-) at (6.5,-2.5) {$\PHat\left(\tau_{m}, S_{\tau_{m}}, \Th^{\tau_{m}}\right)$};

\node [noDraw/.try, neuron /.try] (-) at (-1.6,-9) {$\tau_{i-1}$};
\node [noDraw/.try, neuron /.try] (-) at (1,-9) {$\tau_{i}$};
\node [noDraw/.try, neuron /.try] (-) at (3.5,-9) {$\tau_{i+1}$};
\node [noDraw/.try, neuron /.try] (-) at (6.5,-9) {$\tau_{m}$};
\end{tikzpicture}

\subsection{Optimisation of the algorithm}

In the algorithm,~$n_l$ denotes the number of layers per sub-network,~$n_N$ the number of neurons 
per layer,~$n_B$ the number of batches used to separate the samples,~$N$ the size of each batch and 
{\em epoch} shall denote the number of times all the samples are fed to the training algorithm.
We finally introduce the following loss function that we aim to minimise:
$$
\Lm\left(\wwb, \ddelta, \beta, \gamma\right) := 
\EE\left[\left|\PHat\left(\tau_0, S_{\tau_0},\Th^{\tau_0}\right) - \EE\left[\PHat\left(\tau_0, S_{\tau_0},\Th^{\tau_0}\right)\right]\right|^2\right],
$$
or, in fact, its version on the sample,
\begin{equation}\label{eq:LossFunctionFwd}
\widehat{\Lm}\left(\wwb, \ddelta, \beta, \gamma\right) := 
\frac{1}{N}\sum_{k=1}^{N}\left|\PHat\left(\tau_0,S_{\tau_0},\left(\Th^{\tau_0}\right)^k\right) - \frac{1}{N}\sum_{l=1}^{N}\PHat\left(\tau_0,S_{\tau_0},\left(\Th^{\tau_0}\right)^l\right)\right|^2.
\end{equation}
It represents the variance of the initial price found by backward iterations for each simulated path. 
Since the initial price is unique and deterministic, 
minimising its variance is natural good way to compute the initial price.
Regarding the optimisation itself, we use the Adaptive Moment Estimation method~\cite{ADAM} 
for the first iterations, and switch to the stochastic gradient method, 
with slower but more stable convergence properties.
The different parameters to calibrate are then
\begin{itemize}
\item The weights~$\wwb$ and biases~$\ddelta$ for each layer of each sub-network;
\item The $\beta$ and $\gamma$ in the batch normalisation.
\end{itemize}
The initial price $\PHat_0$ is determined by an average over one or several batches (see~\cite{Yu} for similar loss functions).

\newpage
\section{Numerics: application to the rough Heston model}\label{sec:Numerics}
We develop numerics for the rough Heston model, developed in~\cite{ERChar, ERHedging, EEFR}, which has the following form:
\begin{equation}\label{eq:RoughHeston}
\left\{
  \begin{array}{lr}
\displaystyle \D S_t = S_t \sqrt{V_t} \D W_t,\\
\displaystyle V_t  = V_0 +  \int_0^t  \Kr(t - s) \left[\kappa(\theta - V_s) \D s + \nu\sqrt{V_s} \D Z_s\right], \\
\D \langle W, Z\rangle_t = \rho\, \D t,
  \end{array}
\right.
\end{equation}
where the kernel is defined as 
$\Kr(t) := \frac{1}{\Gamma(\alpha)} t^{\alpha - 1}$, for $\alpha \in (\frac{1}{2},1)$.
By~\cite{AffineVolterra}, the variance process~$V$ admits a non-negative weak solution if both~$\theta$ and~$V_0$ are non-negative.
The model is not Markovian, and a precise Monte Carlo scheme with accurate convergence is not available yet.
We adapt the hybrid scheme from~\cite{BLP15} to the rough Heston model, detailed in Appendix~\ref{appsec:rHestonSimul}.
El Euch and Rosenbaum~\cite{ERChar} showed that, similarly to the standard Heston model ($\alpha=1$), 
the characteristic function of the stock price can be computed in semi-closed form, 
and the next section details this, as well as a numerical scheme, that we use for comparisons.
\subsection{Pricing via fractional Riccati equations}
For $t\geq 0$, $\Phi_t:\RR\to\CC$ denotes the characteristic function
$$
\Phi_t(u) = \EE\left[\left(\frac{S}{S_0}\right)^{\I u}\right].
$$
El Euch and Rosenbaum~\cite{ERChar} showed that 
$\log\Phi_t(u) = \kappa\theta  \Ik^1 h(u, t) + V_0 \Ik^{1 - \alpha} h(u, t)$,
where, for any $u\in\RR$, $h(u, t)$ solves the fractional Riccati equation
\begin{equation}\label{eq:FRE}
\Dk^{\alpha} h(u, t) = \Fk(u, h(u, t)),
\end{equation}
with boundary condition $\Ik^{1-\alpha}h(u,0)=0$ and 
$$
\Fk(u, x) := -\frac{u(u+\I)}{2} + (\I u \rho  \nu - \kappa)  x + \frac{\nu^2 x^2}{2}.
$$
Here, $\Ik^r$ and $\Dk^r$ denote the fractional integral and the fractional derivative defined by
$$
\Ik^r f(t) := \frac{1}{\Gamma(r)}  \int_0^t (t - s)^{r - 1} f(s) \D s
\qquad\text{and}\qquad
\Dk^r f(t) := \frac{1}{\Gamma(1 - r)} \frac{\D}{\D t} \int_0^t (t - s)^{-r} f(s) \D s.
$$
Contrary to the Heston model, the fractional equation~\eqref{eq:FRE} has no closed-form expression, 
and the Adams scheme was proposed in~\cite{ERChar} to compute it.
Taking fractional integrals of both sides of~\eqref{eq:FRE} yields
$$
h(u, t) = \frac{1}{\Gamma(\alpha)} \int_0^t (t - s)^{\alpha - 1} \Fk\left(u, h(u, s) \right) \D s.
$$
We now define an equidistant grid $(t_k)_{0 \le k \le n}$ with mesh size $\Delta$ such that $t_k = k \Delta$.
For each $k = 0, \ldots, n$,  we approximate~$h(\cdot, \cdot)$ by 
$\widehat{h}(u, t_0) := 0$ and 
$$
\widehat{h}(u, t_k) := \frac{1}{\Gamma(\alpha)} \int_0^{t_k} (t_k - s)^{\alpha - 1} \widehat{g}(u, s) \D s, \qquad \text{for }k=1, \ldots, n,
$$
where, for $t \in [t_j, t_{j + 1})$, with $0\leq j\leq k-1$, 
$$
\widehat{g}(u, t) := \frac{t_{j+1} - t}{t_{j + 1} - t_j} \Fk\left(u, \widehat{h}(u, t_j)\right)
 + \frac{t - t_j}{t_{j + 1} - t_j} \Fk\left(u, \widehat{h}(u, t_{j+1})\right)
$$
is a linear interpolation function.
Therefore
\begin{equation}\label{eq:Hhat}
\widehat{h}(u, t_k) = \sum_{j = 0}^{k - 1} a_{j, k} \Fk\left(u, \widehat{h}(u, t_j)\right) + a_{k, k} \Fk\left(u, \widehat{h}(u, t_k)\right),
\end{equation}
with
\begin{equation*}
\left\{
  \begin{array}{lr}
\displaystyle    a_{0, k} = \frac{\Delta^{\alpha}}{\Gamma(\alpha + 2)} \left((k - 1)^{\alpha+1} - (k - \alpha - 1) k^{\alpha} \right),\\
\displaystyle    a_{j, k} =  \frac{\Delta^{\alpha}}{\Gamma(\alpha + 2)} \left[(k - j + 1)^{\alpha + 1} + (k - j - 1)^{\alpha + 1} - 2 (k - j)^{\alpha + 1} \right], 1 \le j \le k  -1,\\
\displaystyle	a_{k, k} = \frac{\Delta^{\alpha}}{\Gamma(\alpha + 2)}.
  \end{array}
\right.
\end{equation*}
Since~$\widehat{h}(u, t_k)$ appears on both sides of~\eqref{eq:Hhat}, it is an implicit scheme which we approximate by an explicit scheme.
Consider a predictor~$\widehat{h}^P(u, t_k)$ derived from the Riemann sum approximating~$\widehat{h}(a, t_k)$ by the integral
$$
\widehat{h}^P(u, t_k) = \frac{1}{\Gamma(\alpha)} \int_0^{t_k} (t_k - s)^{\alpha - 1} \widetilde{g}(u, s) \D s,
$$
with $\widetilde{g}(u, t) := \widehat{g}(u, t_j)$  for $t \in [t_j, t_{j + 1})$, and therefore
$$
\widehat{h}^P(u, t_k) = \sum_{j = 0}^{k - 1} b_{j, k} \Fk\left(u, \widehat{h}(u, t_j)\right)
\qquad\text{with} \qquad
b_{j, k} = \frac{\Delta^{\alpha}}{\Gamma(\alpha + 1)} \left[(k - j)^{\alpha} - (k - j - 1)^{\alpha}\right].
$$
The final scheme therefore reads
$$
\widehat{h}(u, t_k) = \sum_{j = 0}^{k - 1} a_{j, k} \Fk\left(u, \widehat{h}(u, t_j)\right) + a_{k, k} \Fk\left(u, \widehat{h}^P(u, t_k)\right).
$$
Lewis~\cite{Lewis} showed that we can recover Call option prices via inverse Fourier transform as
$$
C(S, T, K) = S - \frac{\sqrt{SK}}{\pi} \int_0^{\infty} \Re\left(\E^{\I u k} \Phi_{T}\left(u - \frac{\I}{2}\right)\right) \frac{\D u}{u^2 + \frac{1}{4}}.
$$

\subsection{Analysis of the algorithm}
We consider the following computer and software specifications:
$$\text{Intel Core i7-6600U, CPU 2.60GHz, 32GB RAM,\quad
Python  3.6.1,\quad
Anaconda 4.4.0,\quad
Tensorflow 1.5.0.
}
$$
and all times below are indicated in seconds.
\subsubsection{Discussion on the number of required networks}
We consider the following parameters for the rough Heston model, without interest rate nor dividend:
\begin{equation}\label{eq:Params}
\kappa = 1, \qquad \nu=0.1, \qquad \alpha = 0.6, \qquad \rho=-0.7, \qquad V_0=0.04, \qquad \theta=0.06, \qquad S_0= 1,
\end{equation}
as well the following network configuration:
we consider $50,000$ Monte Carlo paths, with $200$ time steps.
For the deep learning algorithm, each neural network consists of~$3$ layers with~$5$ neurons each,
the learning rate is set to~$0.2$ and the number of iterations set to~$1000$.
In this example, we discretise the curve~$\Theta$ on a space of functions of dimension~$10$.
We consider~$20$ log-moneynesses ranging from $-0.4$ to~$0.4$, for maturities in $\{0.1, 0.5, 1.6, 5.0\}$ (expressed in years).
The loss function we consider takes into account all strikes for each maturity, and not each single strike individually,
therefore increasing the speed of the algorithm.
Table~\ref{tab:nBSDE} below shows the price errors between the BSDE scheme and the Riccati method,
for different number~$m$ of BSDE time steps, corresponding in fact to different numbers of neural networks.
Surprisingly at first, increasing the number of BSDE time steps ($m$)--clearly more computationally intensive--does not improve the accuracy. 
This can be explained by the fact that a larger number of networks implies more parameters to optimise over,
and therefore reduces accuracy.
This leads us to advocate an algorithm between the one by E, Han, Jentzen~\cite{JentzenDeep},
where the number of networks is equal to the number of Monte Carlo time steps,
and the one by Chan-Wai-Nam, Mikael, Warin~\cite{Nam}, who consider a single network.
There, \emph{Build time} corresponds to the time to build the network, \emph{train time} is the training time,
\emph{DL~Mean~|Error|} is the mean absolute error (across all strikes) between the deep learning price and the Riccati price (computed with~$100$ discretisation steps), 
and \emph{DL~Max~|Error|} is the max absolute errors across all strikes.
As a comparison, the errors and computation times between the Monte Carlo, the Riccati and the Deep Learning prices (with $m=4$) are detailed in Tables~\ref{tab:CompareRicMC_times}-\ref{tab:CompareRicMC}, 
where \emph{MC~Mean~|Error|} and \emph{MC~Max~|Error|} have analogous meanings,
and the detailed prices are given in Figures~\ref{tab:RefPrice_T01}-\ref{tab:RefPrice_T05}-\ref{tab:RefPrice_T16}-\ref{tab:RefPrice_T5}
to help interpret the quantities.
There, we only plot one figure (on the left) for the prices (the Monte Carlo prices) as the others (Riccati and DL) are indistinguishable.

\begin{table}[!htb]
\centering
\begin{tabular}{ccccccc}
\hline
\rowcolor{Gray}
 &  & T & Build time & Train time& DL Mean |Error| & DL Max |Error|\\
\hline
\multirow{4}{*}{\rotatebox{90}{$m=2$}} 
& & 0.1& 31.31 & 24.28 & 3.2E-4 & 2.2E-3\\
& & 0.5& 26.10 & 25.49 & 5E-6 & 2E-4\\
& & 1.6& 26.78 & 21.55 & 2.7E-3 & 4.4E-3\\
& & 5& 25.91 & 20.52 & 2.4E-2 & 2.7E-2\\
\hline
\multirow{4}{*}{\rotatebox{90}{$m=4$}} 
& & 0.1& 43.48 & 23.30 & 4.5E-4 & 2.3E-3\\
& & 0.5& 43.36 & 23.98 & 9.8E-4 & 1.2E-3\\
& & 1.6& 53.54 & 33.87 & 5.9E-3 & 7.4E-3\\
& & 5& 47.12 & 26.69 & 3.9E-2 & 4.1E-2\\
\hline
\multirow{4}{*}{\rotatebox{90}{$m=5$}} 
& & 0.1& 46.46 & 25.77 & 3.4E-4 & 2.2E-3\\
& & 0.5& 44.12 & 30.70 & 7.3E-4 & 9E-4\\
& & 1.6& 40.72 & 21.98 & 1.5E-3 & 2.9E-3\\
& & 5& 44.25 & 21.03 & 6.8E-3 & 8.6E-3\\
\hline
\multirow{4}{*}{\rotatebox{90}{$m=10$}} 
& & 0.1& 64.14 & 21.06 & 2.8E-4 & 2.2E-3\\
& & 0.5& 65.81 & 20.38 & 1.5E-2 & 1.6E-2\\
& & 1.6& 72.67 & 23.14 & 2.3E-2 & 2.5E-2\\
& & 5& 68.30 & 21.04 & 4.6E-2 & 4.8E-2\\
\hline
\multirow{4}{*}{\rotatebox{90}{$m=20$}} 
& & 0.1& 132.06 & 24.247  & 2.3E-2 & 2.5E-2\\
& & 0.5& 132.24 & 27.149  & 5.2E-2 & 5.2E-2\\
& & 1.6& 133.12 & 22.513  & 1.3E-1 & 1.4E-1\\
& & 5& 124.92 & 20.067  & 1E-1 & 1E-1\\
\hline
\end{tabular}
\caption{Price errors computed via the deep learning algorithm with different ($m$) BSDE steps}
\label{tab:nBSDE}
\end{table}

\begin{table}[!htb]
\centering
\begin{tabular}{c|cccccc}
\hline
\rowcolor{Gray}
& MC time & Riccati time\\
\hline
T=0.1& 133 & 54\\
T=0.5& 138 & 57\\
T=1.6& 138 & 54\\
T=5& 140 & 55\\
\hline
\end{tabular}
\caption{Monte Carlo and Riccati computation times (in seconds)}
\label{tab:CompareRicMC_times}
\end{table}

\begin{table}[!htb]
\centering
\begin{tabular}{c|cccccc}
\hline
\rowcolor{Gray}
& MC Mean |Error| & MC Max |Error| & DL Mean |Error| & DL Max |Error|\\
\hline
T=0.1 & 5E-6 & 2.2E-3 & 2.5E-5 & 2.2E-2\\
T=0.5 & 1.4E-4 & 4E-4 & 1.8E-4 & 5E-4\\
T=1.6 & 1.2E-3 & 2.7E-3 & 9.1E-4 & 1.6E-2\\
T=5 & 3.7E-3 & 5.6E-3 & 6.8E-4 & 3.8E-2\\
\hline
\end{tabular}
\caption{Monte Carlo and DL (with $m=4$) price errors}
\label{tab:CompareRicMC}
\end{table}

\begin{figure}
\centering
\includegraphics[scale=0.5]{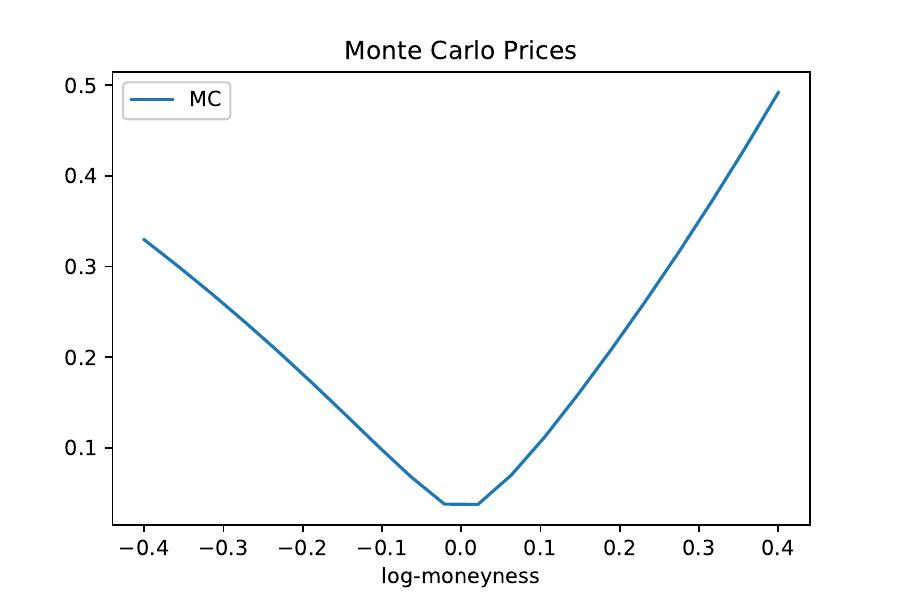}
\includegraphics[scale=0.5]{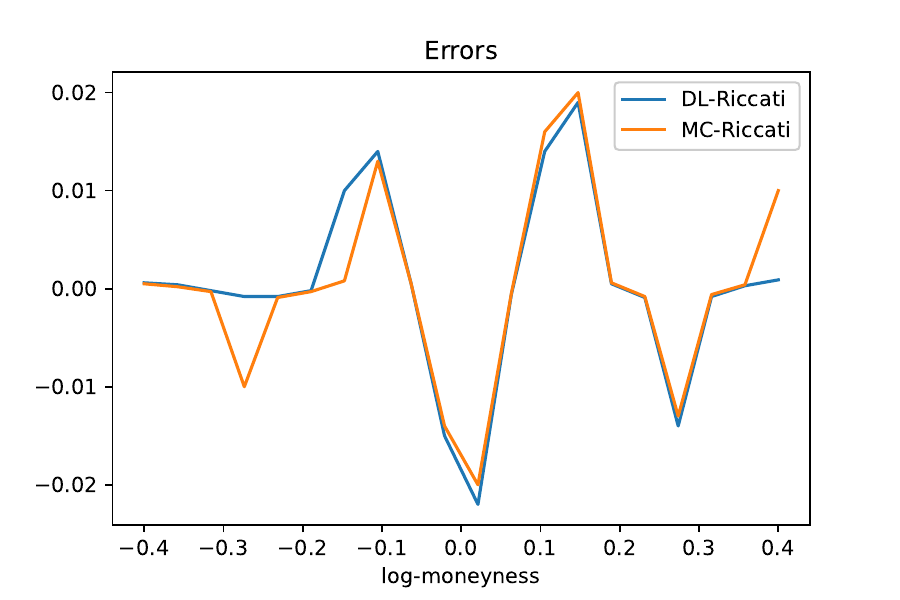}
\caption{Prices for $T = 0.1$}
\label{tab:RefPrice_T01}
\end{figure}

\begin{figure}
\centering
\includegraphics[scale=0.5]{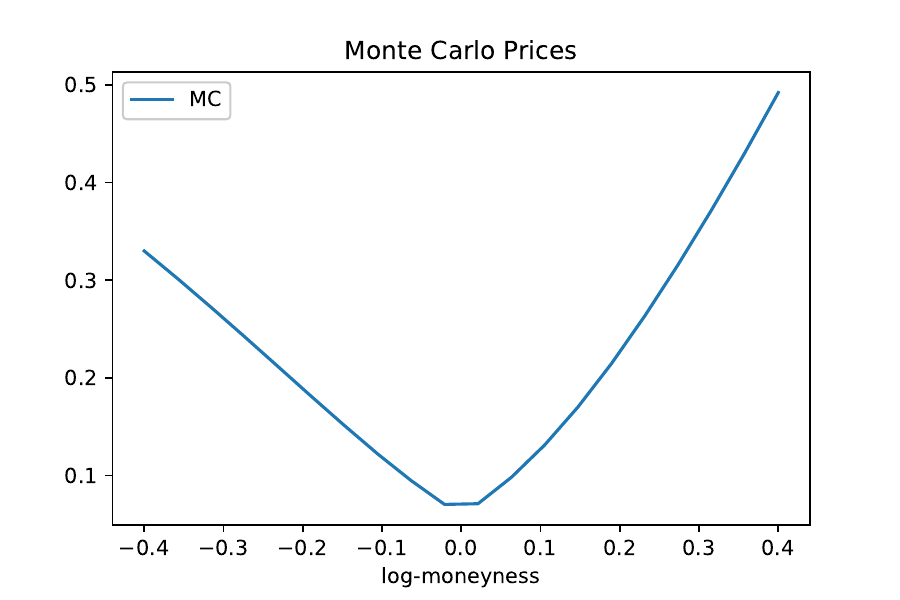}
\includegraphics[scale=0.5]{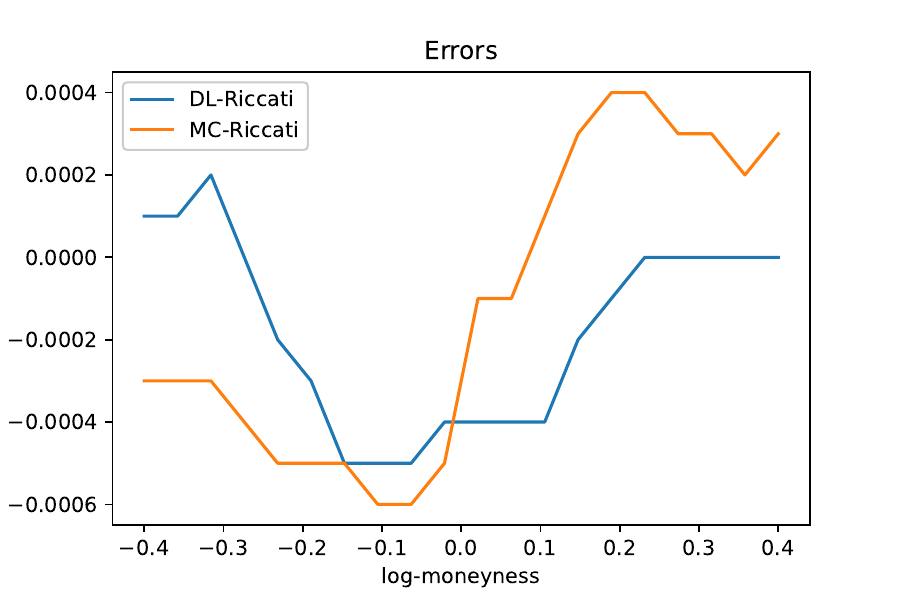}
\caption{Prices for $T = 0.5$}
\label{tab:RefPrice_T05}
\end{figure}

\begin{figure}
\centering
\includegraphics[scale=0.5]{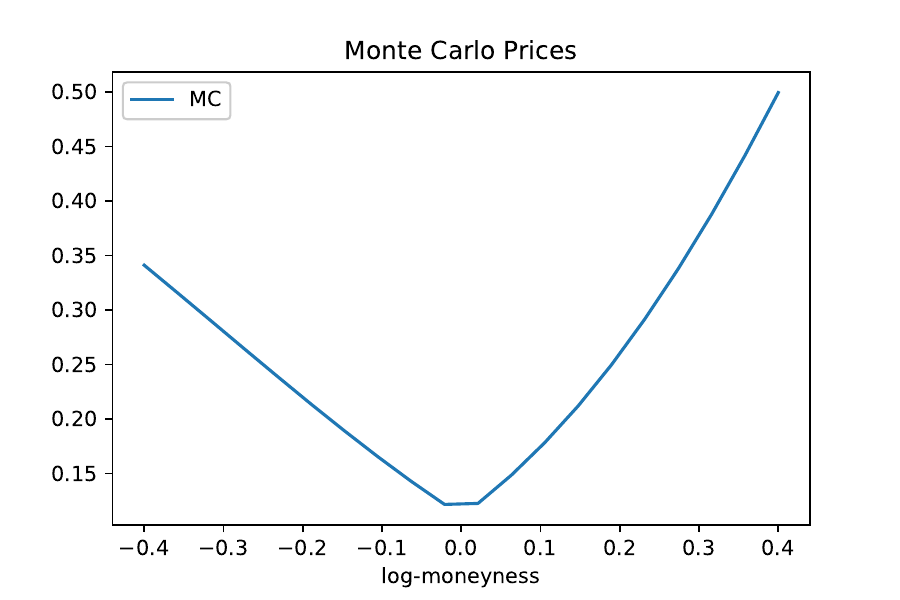}
\includegraphics[scale=0.5]{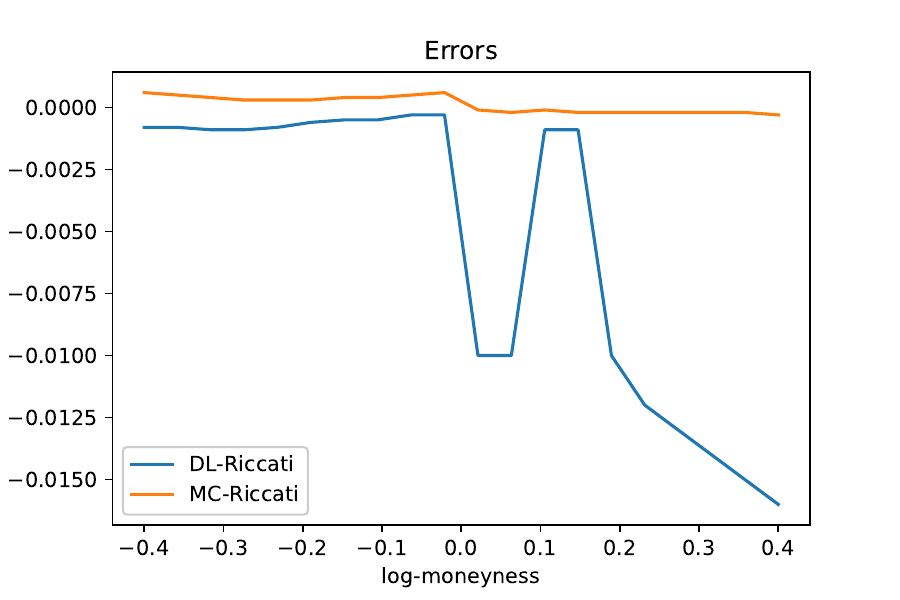}
\caption{Prices for $T = 1.6$}
\label{tab:RefPrice_T16}
\end{figure}

\begin{figure}
\centering
\includegraphics[scale=0.5]{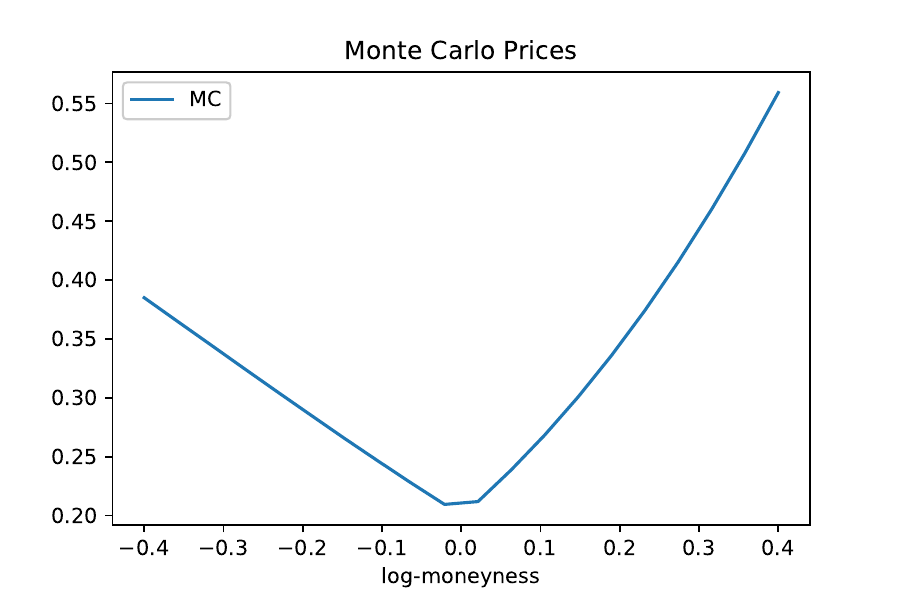}
\includegraphics[scale=0.5]{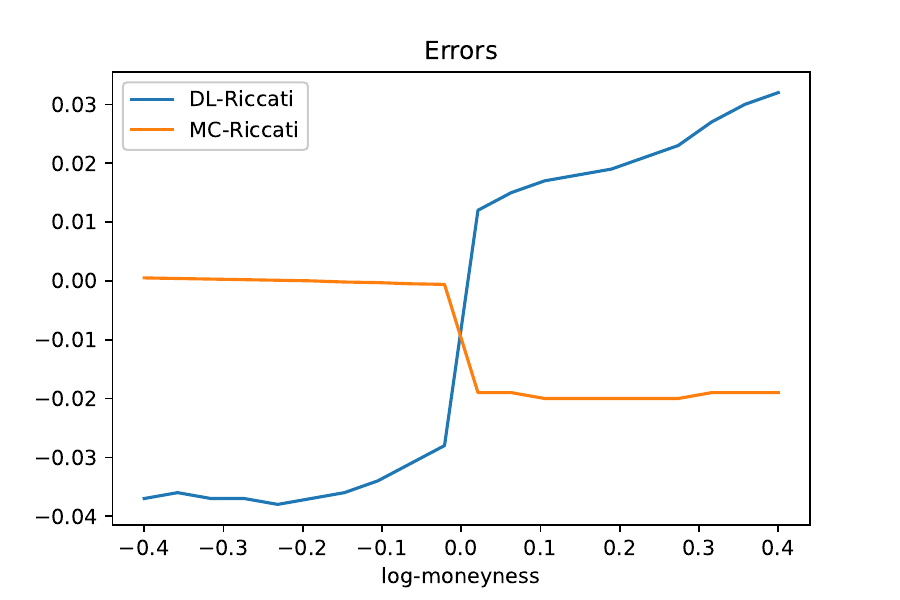}
\caption{Prices for $T = 5$}
\label{tab:RefPrice_T5}
\end{figure}

\newpage
\subsubsection{Shapes of the smile}
We now consider the following set of parameters (we have bumped the volatility of volatility and the correlation on purpose to capture the level of the skew on Equity markets):
\begin{equation}\label{eq:Params2}
\kappa = 1, \qquad \nu=0.9, \qquad \alpha = 0.6, \qquad \rho=-0.8, \qquad V_0=0.04, \qquad \theta=0.06, \qquad S_0= 1.
\end{equation}
We shall consider two maturities, and, with the same hyper-parameters as above, amended with the configurations given in Table~\ref{tab:Config}
(\textit{Networks} is the number of BSDE steps,
\textit{Neurons} stands for the number of neurons per layer, 
and \textit{Layers} is the number of layers).
The implied volatility smiles are plotted in in Figure~\ref{fig:Smiles06}, 
where we added the smiles computed with the Adams scheme for the Riccati equation with~$400$ discretisation steps. 
It is well known that the Riccati version is not so accurate for small maturities and very steep skews as in our example here 
(with the parameters in~\eqref{eq:Params2}), and one could for example use the 
power series algorithm proposed by Callegaro, Grasselli and Pag\`es~\cite{Callegaro}.
The conclusion here is that one should overall be parsimonious with the number of hyper-parameters:
increasing the number of layers or the number of neurons per layer is far from optimal, 
and a few of them are enough for sufficient accuracy.
We leave for future research and deeper numerical analysis the study and convergence of such an algorithm 
for path-dependent options.
In this case, we believe that the flexibility of our network compared to the competitors is key, 
as one can match the BSDE time steps (hence number of networks) with the path-dependent constraints of the setup:
early exercise features for American options, no-exercise periods for convertible bonds for example.

\begin{table}[ht!]
\centering
\begin{tabular}{c|cccccc}
\hline
\rowcolor{Gray}
 &  & Config ID & Networks (m) & Neurons & Layers \\
\hline
\multirow{4}{*}{\rotatebox{90}{$T=0.6$}} 
& & 1 & 4 & 6 & 4\\
& & 2 & 4 & 3 & 2\\ 
& & 3 & 2 & 3 & 4\\ 
& & 4 & 4 & 10 & 4\\ 
\hline
\multirow{4}{*}{\rotatebox{90}{$T=0.2$}} 
& & 1 & 5 & 6 & 4\\
& & 2 & 3 & 10 & 2\\ 
& & 3 & 3 & 6 & 6\\ 
& & 4 & 2 & 6 & 4\\ 
\hline
\end{tabular}
\caption{List of configurations}
\label{tab:Config}
\end{table}

\begin{figure}[h!]
  \includegraphics[scale=0.4]{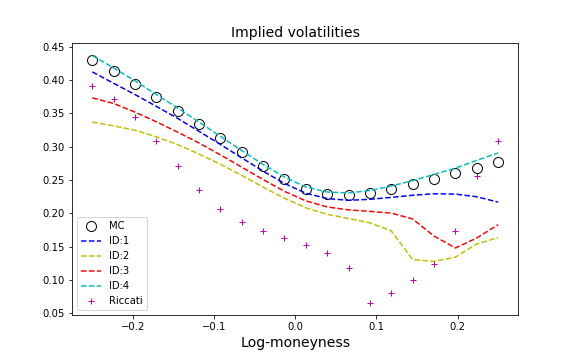}
  \includegraphics[scale=0.4]{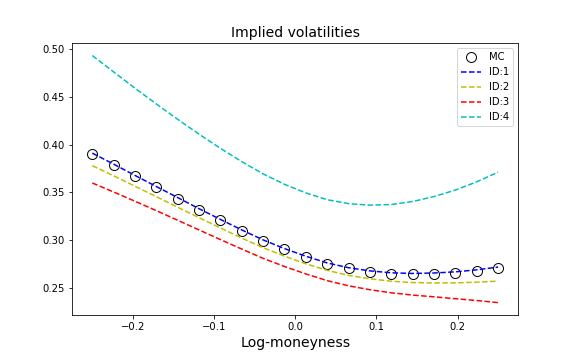}
  \caption{Implied volatility smiles for different configurations for $T=0.2$ (left) and $T=0.6$ (right).}
  \label{fig:Smiles06}
\end{figure}

\begin{remark}
As Figure~\ref{fig:Smiles06} shows, it is not easy to specify the optimal configuration (number of layers and number of nodes) for the neural network.
With enough computing power, the optimisation would perform strongly, 
so one would be inclined to add more nodes and layers. 
However, the dimension of the optimisation problem grows and may yield to multiple optima, 
which each may be different depending on the maturity of the option.
This is a standard problem with neural networks and a common practical recommendation is to keep a low number of layers (such as Config~ID:~3) and increase the number of nodes.
One could (should) add an extra layer in order to make this more robust, such as pruning or growing the network or using an adaptive framework as in~\cite{AdaNet}.
We leave this additional layer to future research and numerical tests.
\end{remark}

\newpage
\appendix

\section{Functional It\^o formula for stochastic Volterra systems}\label{sec:Review}
We gather here several results from Viens and Zhang~\cite{Viens}, which are key to our analysis.
\subsection{Stochastic Volterra systems}
We consider here a stochastic Volterra process of the form
\begin{equation}\label{eq:SDE}
\Xm_t = \Xm_0 + \int_{0}^{t}\bm(t,r,\Xm_{\cdot})\D r + \int_{0}^{t}\sm(t,r,\Xm_{\cdot})\cdot\D \Wm_r,
\end{equation}
where $\Xm_0 \in \RR^d$, $\Wm$ is a Brownian motion in~$\RR^n$
on a given filtered probability space $(\Omega, \Ff, (\Ff_t)_{t\geq 0}, \PP)$,
and~$\bm:\RR_+\times\RR_+\times\RR^d\to\RR^d$ and $\sm:\RR_+\times\RR_+\times\RR^d\to\RR^{d\times n}$ satisfy the following assumptions:
\begin{assumption}\label{assu}
The processes~$\bm$ and~$\sm$ are adapted to~$(\Ff_t)_{t\geq 0}$,
and the derivatives~$\partial_{t}\bm$ and~$\partial_{t}\sm$ exist.
Furthermore, for $\varpb\in\{\bm, \sm, \partial_{t}\bm, \partial_{t}\sm\}$,
$|\varpb(t,s,\omega)| \leq C\left(1+\|\omega\|_T^a\right)$ for some $C, a>0$.
\end{assumption}
Here, $\|\omega\|_T := \sup_{0\leq t \leq T}|\omega_t|$ denotes the supremum norm
on the interval~$[0, T]$.
Saying that~$\varpb$ is adapted here is equivalent to the fact that it can be written as 
$\varpb(t,r,\Xm_{\cdot}) = \varpb(t,r,\Xm_{r\wedge\cdot})$.
The following assumption ensures that the system~\eqref{eq:SDE} is well defined in the following sense:
\begin{assumption}\label{assu2}
The SDE~\eqref{eq:SDE} admits a weak solution and 
$\EE\left[\sup_{t\in[0,T]}|\Xm_t|^p\right]$ is finite for any~$p\geq 1$.
\end{assumption}
This assumption follows~\cite{Viens}.
We do not require strong solutions, as the noise~$W$ is not observable, 
and only~$\Xm$ is (or at least some of its components, for example~$S$ in the rough Heston model).
We refer the reader to~\cite{Berger1, Berger2} for precise conditions
on~$\bm$ and~$\sm$ ensuring weak existence of a solution.
The moment condition is more technical and needed for the functional It\^o formula in 
Theorem~\ref{thm:Ito}.
Most interesting models in the finance literature satisfy these assumptions, 
and we refer to~\cite[Appendix]{Viens} for sufficient conditions
ensuring Assumption~\ref{assu2}, 
in particular for the class of rough affine models~\cite{AffineVolterra}.
The key differences between~\eqref{eq:SDE} and a classical stochastic differential equation is that both drift and diffusion  depend
(a) on two time variables (thus violating the flow property), 
and~(b) on the whole path.
The other classical issue is that the coefficients may blow up, 
as for the Riemann-Liouville fractional Brownian motion
$\int_{0}^{t}(t-s)^{H-1/2}\D W_s$,
where the power-law kernel explodes on the diagonal whenever the Hurst exponent~$H$ lies in $(0, \frac{1}{2})$.
Following the terminology introduced in~\cite{Viens}, two cases have to be distinguished:
\begin{assumption}\label{assu:Kernel}\ 
\begin{enumerate}[(i)]
\item (Regular case) 
For any $s \in [0,T]$, $\dt \bm(t, s, \cdot)$ and $\dt \sm(t, s, \cdot)$ exist on $[s, T]$,
and for $\varpb\in\{\bm, \sm, \dt \bm, \dt \sm\}$,
$$
\left|\varpb(t,s,\omega)\right|\leq C\left(1+\|\omega\|_{T}^a\right),
\qquad\text{for some }a, C>0;
$$
\item (Singular case) 
Let $\varpb\in\{\bm, \sm\}$.
For any $s \in [0,T]$, $\dt \varpb(t, s, \cdot)$ exists on $(s, T]$,
and there exists $h\in\left(0, \frac{1}{2}\right)$ such that, for some $a, C>0$,
$$
\left|\varpb(t,s,\omega)\right|\leq C\left(1+\|\omega\|_{T}^a\right)(t-s)^{h-1/2}
\qquad\text{and}\qquad
\left|\dt\varpb(t,s,\omega)\right|\leq C\left(1+\|\omega\|_{T}^a\right)(t-s)^{h-3/2}.
$$
\end{enumerate}
\end{assumption}
The first case mainly deals with the path dependence and the absence of the Markov property, 
while the second one allows us to treat the presence of two time variables in the kernel, 
which occurs in fractional models, and in particular in the setting of Section~\ref{sec:Modelling} above.
For any $0\leq t \leq u$, we can decompose~\eqref{eq:SDE} as
\begin{equation}\label{eq:ThetaViens}
\Xm_u 
 = \underbrace{\Xm_0 + \int_{0}^{t}\bm(u,r,\Xm_{r\wedge\cdot})\D r + \int_{0}^{t}\sm(u,r,\Xm_{r\wedge\cdot})\D \Wm_r}_{\displaystyle \Theta_{u}^{t}\in\Ff_t}
 + \underbrace{\int_{t}^{u}\bm(s,r,\Xm_{r\wedge\cdot})\D r + \int_{t}^{u}\sm(u,r,\Xm_{r\wedge\cdot})\D \Wm_r}_{\displaystyle I_{u}^{t}\notin\Ff_t}.
\end{equation}
We further recall (from~\cite{Viens}) the concatenation notation of the paths~$\Xm$ 
and~$\Theta^t$ before and after time~$t$,
\begin{equation}\label{eq:Concat}
\left(\Xm\ot\Theta^t\right)_{u} := \Xm_u\,\ind_{\{0<u<t\}} + \Theta^t_u\,\ind_{\{t<u<T\}},
\qquad\text{for any }u, t\in [0, T].
\end{equation}

\subsection{Functional It\^o calculus}\label{sec:FuncIto}
For any $t \in [0,T]$, let~$\Dd_t$ and ~$\Cc_t$ denote respectively 
the space of c\`adl\`ag functions on $[t, T]$ and that of continuous functions on $[t, T]$, as well as 
$$
\Lo := \left\{(t, \omega) \in [0,T] \times \Dd_0: \omega_{[t,T]} \in \Cc_t \right\}
\qquad\text{and}\qquad
\Lambda := [0, T]\times \Cc\left([0,T], \RR^d\right),
$$
where~$\omega_{[t,T]}$ refers to the truncation of the path~$\omega$ to the interval $[t, T]$.
We denote by~$\Cc(\Lo)$ the space of all functions on~$\Lo$, continuous with respect to
the distance function
$\dist((t, \omega), (t', \omega')) : = |t-t'| + \|\omega-\omega'\|_T$.
For a given $u\in \Cc(\Lo)$, we define its (right) time derivative as
$$
\dt u(t, \omega) := \lim_{\eps\downarrow 0}\frac{u(t+\eps, \omega) - u(t, \omega)}{\eps},
\qquad\text{for all } (t, \omega) \in \Lo.
$$
Following~\cite{Viens}, we then define spatial derivatives of~$u\in\Cc(\Lo)$ as linear or bilinear operators on~$\Cc_t$:
\begin{definition}\label{def:FuncDeriv}
The spatial derivatives of~$u\in\Cc(\Lo)$ are defined as Fr\'echet derivatives.
For any $(t, \omega)\in\Lo$, 
\begin{equation*}
\begin{array}{rll}
\la \dom u(t, \omega), \eta\ra
&  := \displaystyle \lim_{\eps\downarrow 0}\frac{u(t, \omega + \eps \eta_{[t,T]}) - u(t, \omega)}{\eps},
 & \text{for any }\eta\in\Cc_t,\\
\la\domom u(t, \omega), (\eta,\zeta)\ra
&  := \displaystyle \lim_{\eps\downarrow 0}\frac{\la\dom u(t, \omega + \eps \eta_{[t,T]}), \zeta\ra
 - \la\dom u(t, \omega), \zeta\ra}{\eps},
 & \text{for any }\eta, \zeta \in\Cc_t.
\end{array}
\end{equation*}
\end{definition}
This definition of the spatial derivative in the direction $\eta\in\Cc_t$ is obviously equivalent to
$$
\la\dom u(t, \omega), \eta\ra
 = \frac{\D}{\D \eps} \left.u\left(t, \omega + \eps \eta_{[t,T]}\right)\right|_{\eps=0}.
$$
This definition is consistent with that of Dupire~\cite{Dupire}, as the perturbation
acts on the time interval $[t, T]$, but not on $[0,t]$,
and the distance function~$\dist(\cdot)$ is similar to Dupire's pseudo-distance (see also~\cite{RenTZ}).
We shall further need the following two spaces:
\begin{align*}
\Cc^{1,2}(\Lo) & := \left\{u \in \Cc(\Lo): \varpb \in \Cc(\Lo) \text{ for }
\varpb \in \{\partial_t u, \partial_{\omega}u, \partial_{\omega}^{2}u\} \right\},\\
\Cc^{1,2}_{+}(\Lo) & := \left\{u \in \Cc^{1,2}(\Lo): 
\varpb \text{ has polynomial growth}  \text{ for }
\varpb \in \{\partial_t u, \partial_{\omega}u, \partial_{\omega}^{2}u\}\right.\\
& \qquad\left.\text{and }
\la\domom u, (\eta,\eta)\ra
\text{ is locally uniformly continuous in~$\omega$ with polynomial growth} \right\}.
\end{align*}
The definition of polynomial growth here is as follows:
\begin{definition}[Definition 3.3 in~\cite{Viens}]
Let $u\in\Cc(\Lo)$ such that~$\dom u$ is well defined on~$\Lo$.
The functional~$\dom u$ is said to have polynomial growth if 
$$
\left|\la \dom u(t, \omega), \eta\ra\right| \leq C\left(1+\|\omega\|_T^\alpha\right)\|\eta_{[t,T]}\|_T,
\qquad\text{for all }(t, \omega)\in\Lo, \eta\in\Cc_t,
$$
for some $C, \alpha>0$.
It is continuous if 
$\Lo\ni (t, \omega)\mapsto \la \dom u(t, \omega), \eta\ra$ is continuous under~$\dist$ for every $\eta\in\Cc$.
\end{definition}

We now recall the main result by Viens and Zhang~\cite[Theorem 3.10 and Theorem 3.17]{Viens},
extending the It\^o formula to the stochastic Volterra framework, 
for both regular and singular cases.
The issue with the singular case (Definition~\ref{assu:Kernel}(ii)) is 
that the coefficients~$\bm$ and~$\sm$ do not belong to~$\Cc_t$ any longer,
so that the Fr\'echet derivatives in Definition~\ref{def:FuncDeriv} do not make sense any more.
In order to develop an It\^o formula, those need to be amended.
We refer the reader to~\cite[Definition 3.16]{Viens} for a precise definition of the space~$\Cc_{+\beta}^{1,2}(\Lo)$, where $\beta \in (0,1)$ intuitively monitors the rate of explosion on the (time) diagonal.

\begin{theorem}\label{thm:Ito}
For $t \in [0, T]$, define $\Zm^t := \Xm\ot\Theta^t$,
and let $\varpb^{t, \omega} := \varpb(\cdot, t, \omega)$ for $\varpb\in\{\bm, \sm\}$
to emphasise the time dependence of the coefficients.
Under Assumptions~\ref{assu}-\ref{assu2}, the following It\^o formula holds:
\begin{equation*}
\begin{array}{rl}
\D u\left(t, \Zm^t\right)
 =  & \displaystyle \left(\dt u\left(t, \Zm^t\right)  + \la \dom u\left(t,\Zm^t\right), \bm^{t,\Xm}\ra
  + \frac{1}{2}\la\domom u\left(t,\Zm^t\right), \left(\sm^{t,\Xm}, \sm^{t,\Xm}\right)\ra\right)\D t\\
 &  \displaystyle + \la\dom u\left(t,\Zm^t\right), \sm^{t,\Xm}\ra\D \Wm_t,
\end{array}
\end{equation*}
\begin{enumerate}
\item in the regular case (Assumption~\ref{assu:Kernel}(i)), whenever $u\in\Cc_+^{1,2}(\Lambda)$;
\item in the singular case (Assumption~\ref{assu:Kernel}(ii))
for $u\in\Cc_{+, \beta}^{1,2}(\Lambda)$ with $\beta+h-\frac{1}{2}>0$, 
where the spatial derivatives should be understood in the regularised sense:
$$
\la \dom u(t, \omega), \phi\ra := \lim_{\delta\downarrow 0}\la \dom u(t, \omega), \phi^{\delta}\ra
\qquad\text{and}\qquad
\la \domom u(t, \omega), (\phi, \phi)\ra := \lim_{\delta\downarrow 0}\la \domom u(t, \omega), (\phi^{\delta}, \phi^{\delta})\ra,
$$
with  the truncated function $\phi^\delta(t,s, \omega) := \phi\left(t\vee (s+\delta), s, \omega\right)$.
\end{enumerate}
\end{theorem}
In the theorem, we invoked the space $\Cc_+^{1,2}(\Lambda)$, which represents
the space of functions $u:\Lambda\to\RR$ such that 
there exists $v\in \Cc_+^{1,2}(\Lo)$ for which $v=u$ on~$\Lambda$.
The derivatives are defined similarly as restrictions on~$\Lambda$.

\section{Proof of Theorem~\ref{thm:PDEBSDE}}\label{sec:ProofThmPDEBSDE}
Since the curve-dependent PDE in the theorem is linear,
existence and uniqueness of the solution is well known~\cite{Touzi3}.
We consider a self-financing portfolio~$\Pi$ consisting 
of the derivative~$P$ given in~\eqref{eq:Price}, 
some quantity~$\Delta$ of stock and some other derivative~$\Psi$, i.e.
at any time $t \in [0,T]$,
$$
\Pi_t = P_t - \Delta_t S_t - \gamma_t \Psi_t.
$$
From Theorem~\ref{thm:Ito}, 
we can write a functional It\^o formula for the option price using Assumption~\ref{assu:Functional}
under the pricing measure~$\QQ$:
$$
\D P_t  = \D \Ps\left(t, S_t, \TTheta^t\right)
 = \Aa \Ps\D t
   + l(t, \Xm_t) S_t\dx \Ps\D W_t
  + \xi(V_t)\la \dom \Ps, \Kr^t\ra \D B_t,
$$
where again $\Xm_t = (S_t, V_t)$, and 
$$
\Aa \Ps:= \dt \Ps + r_t S_t \dx \Ps + \frac{l(t, \Xm_t)^2}{2}S_t^2\dxx \Ps
 + \frac{\xi(V_t)^2}{2}\la\domom \Ps, (\Kr^t, \Kr^t) \ra 
 + b(V_t)\la \dom \Ps, \Kr^t\ra + l(t, \Xm_t)\rho\xi(V_t)S_t\la \dxom \Ps, \Kr^t\ra.
$$
Here, the $x$-derivative refers to the classical derivative with respect to the second component~$S_t$,
whereas the $\omega$-derivative is the Fr\'echet-type derivative in the direction given by the one-dimensional  path~$\omega$.
Applying directly Theorem~\ref{thm:Ito} with the SDE~\eqref{eq:XSDE} for the process~$\Xm$, 
we should normally obtain terms of the form $\la \dOm\Ps, \bm^t\ra$
(and similarly for the second and the cross derivatives),
where~$\bm^t_r = \bm(r, t, \Xm_r)$ for $r \in [0,T]$ is in~$\RR^2$
and $\Om = (\omega_1, \omega_2)$ the two-dimensional path.
We can then write
\begin{align*}
\la \dOm\Ps\left(t, S_t, \TTheta^t\right), \bm^t\ra
 & = \mu_t S_t\partial_{\omega_1}\Ps\left(t, S_t, \TTheta^t\right)
 + \la \partial_{\omega_2}\Ps\left(t, S_t, \TTheta^t\right), \Kr^t b(V_t)\ra\\
 & = \mu_t S_t\partial_{x}\Ps\left(t, S_t, \TTheta^t\right)
+ b(V_t)\la \partial_{\omega_2}\Ps\left(t, S_t, \TTheta^t\right), \Kr^t \ra.
\end{align*}
This yields, similarly, for the portfolio~$\Pi$, under~$\QQ$,
\begin{align*}
\D \Pi_t & = \D \Ps - \Delta_t \D S_t - \gamma_t \D \Psi_t\\
 & = \Aa \Ps \D t  + l(t, \Xm_t) S_t\dx \Ps\D W_t
  + \xi(V_t)\la \dom \Ps, \Kr^t\ra \D B_t\\
 &  - \Delta_t\Big(\mu_t S_t \D t  + l(t, \Xm_t)S_t\D W_t\Big)
 - \gamma_t \Big(\Aa \Psi_t \D t  + l(t, \Xm_t) S_t\dx \Psi_t\D W_t
  + \xi(V_t)\la \dom \Psi_t, \Kr^t\ra \D B_t\Big).
\end{align*}
The portfolio is risk free if $\D\Pi_t = r_t \Pi_t \D t $ and the random noise is cancelled, meaning that
\begin{equation}\label{eq:RiskFreePi}
\left\{
\begin{array}{rl}
\Big(\Aa \Ps -\gamma_t \Aa \Psi_t - \Delta_t \mu_t S_t\Big)\D t
 & = r_t\Big(\Ps - \Delta_t S_t - \gamma_t \Psi_t\Big)\D t,\\
 \dx \Ps - \gamma_t \dx \Psi_t - \Delta_t & = 0,\\
\la \dom \Ps, \Kr^t\ra - \gamma_t  \la \dom \Psi_t, \Kr^t\ra & = 0,
\end{array}
\right.
\end{equation}
since both functions $l(\cdot, \cdot, \cdot)$ and~$\xi(\cdot)$ are nowhere null.
The last two equalities yield
$$
\gamma_t = \frac{\la \dom \Ps, \Kr^t\ra}{\la \dom \Psi_t, \Kr^t\ra}
 \qquad\text{and}\qquad
\Delta_t = \dx \Ps - \frac{\la \dom \Ps, \Kr^t\ra}{\la \dom \Psi_t, \Kr^t\ra}\dx \Psi_t.
$$
We can now rewrite the first equality in~\eqref{eq:RiskFreePi} as
$$
\Big(\Aa \Ps -\gamma_t \Aa \Psi_t - \Delta_t r_t S_t\Big)\D t
 = r_t\Big(\Ps - \Delta_t S_t - \gamma_t \Psi_t\Big)\D t,
$$
which is equivalent to
$$
\frac{\left(\Aa - r_t\right) \Ps}{\la \dom \Ps, \Kr^t\ra}
 = \frac{\left(\Aa - r_t\right)\Psi_t}{\la \dom \Psi_t, \Kr^t\ra}.
$$
The left-hand side is a function of~$P$ only, whereas the right-hand side only depends on~$\Psi$.
Therefore, the only way for this equality to hold is for both sides to be equal to some function~$-\widehat{b}$ 
that depends on~$S_t$, $\Theta^t$ and~$t$, but not on~$P$ nor~$\Psi$.
The pricing equation for the price function is therefore
$$
\left(\Aa - r_t\right) \Ps = - \la \dom \Ps, \Kr^t\ra \widehat{b}_t.
$$
Following similar computations in classical (Markovian) stochastic volatility models, 
we consider~$\widehat{b}_t$ of the form $\widehat{b}_t = b(V_t) - \xi(V_t)\lambda_t$, 
where~$\lambda_t$ is called the market price of risk.
The final pricing PDE is therefore
$$
\dt + r_t S_t \dx + \frac{l(t,\Xm_t)^2}{2}S_t^2\dxx 
 + \frac{\xi(V_t)^2}{2}\la\domom, (\Kr^t, \Kr^t) \ra 
 + b(V_t)\la \dom, \Kr^t\ra + l(t, \Xm_t)\rho\xi(V_t)S_t\la \dxom, \Kr^t\ra
 + \la \dom , \Kr^t\ra \widehat{b}_t = r_t.
$$
With a slight abuse of notations, writing~$b$ in place of~$\widehat{b}$ proves the statement.


\section{Simulation of the rough Heston model}\label{appsec:rHestonSimul}
We provide here details about the simulation of the rough Heston model in~\eqref{eq:RoughHeston}.
Introducing the infinite-dimensional process $(\Theta^t)_{t\geq 0}$ as above, we can write, for any
$t\geq 0$ and $u\geq t$,
\begin{equation}\label{eq:ThetaRHeston}
\Theta^t_u = V_0 + \int_{0}^{t}\Kr(u-s)\Big[\kappa(\theta-V_s)\D s + \xi \sqrt{V_s}\D B_s\Big].
\end{equation}
Given a fixed time horizon $T>0$ and a given number of time steps~$n$, we introduce an equidistant grid for the closed interval $[0,T]$
as $t_i = i/n$, for $i=0, \ldots, n$.
Discretising the rough SDE for the variance process in~\eqref{eq:RoughHeston} along this grid, and denoting $V_i = V_{t_{i}}$ for simplicity,
we can write $V_0 = V_0$ and, for any $i=1,\ldots, n$,
\begin{align}\label{eq:DiscrVrHeston}
V_i & = V_0 + \int_{0}^{t_i}\Kr(t_i-s)\Big[\kappa(\theta-V_s)\D s + \xi \sqrt{V_s}\D B_s\Big]\nonumber\\
 & \approx V_0 + \sum_{j=0}^{i-1}\kappa(\theta-V_j) \int_{t_j}^{t_{j+1}}\Kr(t_i-s)\D s
 + \sum_{j=0}^{i-1}\xi \sqrt{V_j} \int_{t_j}^{t_{j+1}}\Kr(t_i-s)\D B_s\nonumber\\
 & = V_0 + \sum_{j=0}^{i-1}\kappa(\theta-V_j) A_{j,i}
 + \sum_{j=0}^{i-2}\xi \sqrt{V_j} \int_{t_j}^{t_{j+1}}\Kr(t_i-s)\D B_s
 + \xi \sqrt{V_{i-1}} \int_{t_{i-1}}^{t_{i}}\,\Kr(t_i-s)\D B_s,
\end{align}
where we freeze the variance process on each subinterval to its left-point value, and single out the singular part of the kernel in the last integral.
We also introduced the quantity
$$
A_{j, i} := \int_{t_j}^{t_{j+1}}\Kr(t_i-s)\D s,
\qquad\text{for }i=1,\ldots, n \text{ and } j=0,\ldots, i-1,
$$
which can be pre-computed and stored. 
Note that, for $i=1$, the middle sum in~\eqref{eq:DiscrVrHeston} does not appear.
Following~\cite{BLP15}, we can write the middle term in the discretisation as
\begin{align}\label{eq:Discr}
\sum_{j=0}^{i-2}\xi \sqrt{V_j} \int_{t_j}^{t_{j+1}}\Kr(t_i-s)\D B_s 
  = \sum_{j=0}^{i-2}\xi \sqrt{V_j}\,\Kr\left(\frac{b^*_{i-j}}{n}\right) \int_{t_j}^{t_{j+1}}\D B_s 
&   = \sum_{k=2}^{i}\xi \sqrt{V_{i-k}}\,\Kr\left(\frac{b^*_{k}}{n}\right) \int_{t_{i-k}}^{t_{i-k+1}}\D B_s \nonumber\\
 & =: \sum_{k=2}^{i}\xi \sqrt{V_{i-k}}\,\Kr\left(\frac{b^*_{k}}{n}\right) \overline{B}_{i-k},
\end{align}
with~$b^*_k$ defined in~\cite[Proposition~2.8]{BLP15} and with 
$\overline{B}_{i} : = \int_{t_{i}}^{t_{i+1}}\D B_s$ for $i=0,\ldots,n-1$.
Finally, for the last term in~\eqref{eq:DiscrVrHeston}, where the singularity occurs, we introduce the vector 
$(\widetilde{B}_{i})_{i=0,\ldots,n-1}$ as
$\widetilde{B}_{i} := \int_{t_{i}}^{t_{i+1}}\Kr(t_{i+1}-s)\D B_s$.
In the notations of~\cite{BLP15}, $\widetilde{B}_{i}$ is denoted $\overline{B}_{i,1}$, but we remove the double index here.
For any $i=0,\ldots,n-1$, the couple $(\overline{B}_{i}, \widetilde{B}_{i})$ 
forms a two-dimensional Gaussian vector,
with covariance matrix~$\Sigma$ given by
$$
\Sigma_{11} = \frac{1}{n},\qquad
\Sigma_{22} = \int_{t_{i}}^{t_{i+1}}\Kr(t_{i+1}-s)^2 \D s,\qquad
\Sigma_{12} = \Sigma_{21} = \int_{t_{i}}^{t_{i+1}}\Kr(t_{i+1}-s) \D s.
$$
Summarising, we discretise the variance process as
\begin{equation}\label{eq:DiscrVSumm}
V_i  = V_0 + \sum_{j=0}^{i-1}\kappa(\theta-V_j) A_{j,i}
 + \sum_{k=2}^{i}\xi \sqrt{V_{i-k}}\,\Kr\left(\frac{b^*_{k}}{n}\right) \overline{B}_{i-k}
 +  \xi \sqrt{V_{i-1}}\, \widetilde{B}_{i-1},
\qquad\text{for }i=1,\ldots, n.
\end{equation}
\begin{remark}
For computational purposes, the steps above can be sped up bearing in mind 
that the matrix $(A_{j, i})_{i,j}$ is a strictly upper triangular Toeplitz matrix
and that the last expression on the right-hand side of~\eqref{eq:Discr}  can be computed as a discrete Fourier transform.
\end{remark}

\begin{remark}\label{rem:HybridPowerLaw}
In the power law case $\Kr(t) = t^{H-\frac{1}{2}}$ with $H\in(0,1)$, the expressions above simplify to
$$
A_{j,i} = \frac{1}{H+\frac{1}{2}}\left\{(t_i-t_j)^{H+\frac{1}{2}} - (t_i - t_{j+1})^{H+\frac{1}{2}}\right\},
$$
\begin{equation}\label{eq:SigmaRHeston}
\Sigma = 
\begin{pmatrix}
\displaystyle \frac{1}{n} & \displaystyle \frac{1}{(H+\frac{1}{2}) n^{H+\frac{1}{2}}}\\
\displaystyle \frac{1}{(H+\frac{1}{2}) n^{H+\frac{1}{2}}} & \displaystyle \frac{1}{2H n^{2H}}
\end{pmatrix},
\end{equation}
and, as shown in~\cite[Proposition 2.8]{BLP15}, the coefficients $(b_k^*)$ are explicitly computed as
\begin{equation}\label{eq:bStartRHeston}
b_k^* := \left(\frac{k^{H+\frac{1}{2}} - (k-1)^{H+\frac{1}{2}}}{H+\frac{1}{2}}\right)^{\frac{1}{H-1/2}}.
\end{equation}
In this case, with the uniform grid $t_i=i/n$, denoting $\widetilde{A}_{i-j} := A_{j,i}$, we can rewrite~\eqref{eq:DiscrVSumm} as
$$
V_i  = V_0 + \sum_{j=0}^{i-1}\kappa(\theta-V_j) \widetilde{A}_{i-j}
 + \sum_{k=2}^{i}\xi \sqrt{V_{i-k}}\,\Kr\left(\frac{b^*_{k}}{n}\right) \overline{B}_{i-k}
 +  \xi \sqrt{V_{i-1}}\, \widetilde{B}_{i-1},
\qquad\text{for }i=1,\ldots, n,
$$
and the vector~$(\widetilde{A}_k)_{k=1, \ldots, n}$ reads
$$
\widetilde{A}_{k} = \frac{k^{H+\frac{1}{2}} - (k-1)^{H+\frac{1}{2}}}{\left(H+\frac{1}{2}\right)n^{H+\frac{1}{2}}}.
$$
\end{remark}

Regarding the process~$\Theta$ in~\eqref{eq:ThetaRHeston}, we discretise it analogously as
$\Theta^i_k = V_k$ whenever $k\leq i$ and, for $k>i$,
\begin{align*}
\Theta^i_k & = V_0 + \int_{0}^{t_i}\Kr(t_k-s)\Big[\kappa(\theta-V_s)\D s + \xi \sqrt{V_s}\D B_s\Big]
  = V_0 + \sum_{j=0}^{i-1}\int_{t_j}^{t_{j+1}}\Kr(t_k-s)\Big[\kappa(\theta-V_s)\D s + \xi \sqrt{V_s}\D B_s\Big]\\
 & \approx V_0 + \sum_{j=0}^{i-1}\kappa(\theta-V_j)\int_{t_j}^{t_{j+1}}\Kr(t_k-s)\D s
+ \xi \sum_{j=0}^{i-1}\sqrt{V_j} \int_{t_j}^{t_{j+1}}\Kr(t_k-s)\D B_s\\
 & = V_0 + \sum_{j=0}^{i-1}\kappa(\theta-V_j) A_{j,k}
+ \xi \sum_{j=0}^{i-1}\sqrt{V_j} \int_{t_j}^{t_{j+1}}\Kr(t_k-s)\D B_s\\
 & = V_0 + \sum_{j=0}^{i-1}\kappa(\theta-V_j) A_{j,k}
+ \xi \sum_{j=0}^{i-1}\sqrt{V_j}\,\Kr\left(\frac{b^*_{k-j}}{n}\right) \overline{B}_{j}\\
 & = \left(V_0 + \sum_{j=0}^{i-2}\kappa(\theta-V_j) A_{j,k}
+ \xi \sum_{j=0}^{i-2}\sqrt{V_j}\,\Kr\left(\frac{b^*_{k-j}}{n}\right) \overline{B}_{j}\right)
+ \kappa(\theta-V_{i-1}) A_{i-1,k} + \xi\sqrt{V_{i-1}}\,\Kr\left(\frac{b^*_{k-(i-1)}}{n}\right) \overline{B}_{i-1}\\
 &  =\Theta^{i-1}_{k} + \kappa(\theta-V_{i-1}) A_{i-1,k} + \xi\sqrt{V_{i-1}}\,\Kr\left(\frac{b^*_{k-(i-1)}}{n}\right) \overline{B}_{i-1}.
\end{align*}
For the stock price, starting from $S_{t_0} = S_0$, 
we use the discretised explicit form, for $i=1, \ldots, n$,
\begin{equation}\label{eq:DiscrS}
S_{i} = S_{i-1} \exp\left\{- \frac{V_{i-1}}{2n}  + \sqrt{V_{i-1}}\, W_{i-1} \right\},
\end{equation}
where $W_i := \int_{t_i}^{t_{i+1}}\D W_s$ for some standard Brownian motion~$W$ such that $\D\langle W, B\rangle_t = \rho\, \D t$.
\begin{remark}
The simulation recipe is as follows:
\begin{itemize}
\item Pre-compute the vector $(b^*_k)$ in~\eqref{eq:bStartRHeston} for $k=1,\ldots, n$;
\item Generate three iid Gaussian samples $\left(\mathrm{N}^{1}_i, \mathrm{N}^{2}_i, \mathrm{N}^{3}_i\right)$ for $i=0,\ldots,n-1$;
\item Recalling the covariance matrix~$\Sigma$ in~\eqref{eq:SigmaRHeston}, compute the Gaussian vector $(\overline{B}_i, \widetilde{B}_i, W_i)_{i=0,\ldots,n-1}$ as
$$
\begin{pmatrix}
\overline{B}_i\\
\widetilde{B}_i\\
W_i
\end{pmatrix}
 = 
\begin{pmatrix}
\Sigma_{11} & \Sigma_{12} & \rho/n \\
\Sigma_{12} & \Sigma_{22} & \rho\Sigma_{12}\\
\rho/n & \rho\Sigma_{12} & 1/n
\end{pmatrix}^{1/2}
\begin{pmatrix}
\mathrm{N}^{1}_{i}\\
\mathrm{N}^{2}_{i}\\
\mathrm{N}^{3}_{i}
\end{pmatrix},
$$
since $\EE[W_i\widetilde{B}_i] = \rho \int_{t_i}^{t_{i+1}}\Kr(t_{i+1}-s)\D s = \rho\Sigma_{12}$,
and the square-root understood in the Cholesky sense;
\item Compute $(V_i)_{i=0,\ldots, n}$ using~\eqref{eq:DiscrVSumm} and $(S_i)_{i=0,\ldots, n}$ using~\eqref{eq:DiscrS};
\item For each $i=0,\ldots, n$, generate the discretised curves $(\Theta^i_k)_{k=0,\ldots,n}$.
\end{itemize}
\end{remark}


\end{document}